\documentclass[11pt,fleqn]{article}
\usepackage{amssymb,latexsym,amsmath,amsfonts,graphicx}

    \advance\textheight by \topskip

    \numberwithin{equation}{section}

    \def\Re{{\rm Re \,}}
    \def\Im{{\rm Im \,}}
    \def\Ai{{\rm Ai \,}}

    \def\bigO{{\cal O}}

    \newtheorem{theorem}{Theorem}[section]

    \newtheorem{proposition}[theorem]{Proposition}
    
    \newtheorem{Definition}[theorem]{Definition}
    
    \newtheorem{Remark}[theorem]{Remark}
    \newenvironment{remark}{\begin{Remark}\rm}{\end{Remark}}
    \newtheorem{Example}[theorem]{Example}
    \newenvironment{example}{\begin{Example}\rm}{\end{Example}}
    \newtheorem{Assumptions}[theorem]{Assumptions}

    \newenvironment{proof}%
    {\rm \trivlist \item[\hskip \labelsep{\bf Proof. }]}%
    {\hspace*{\fill}$\Box$\endtrivlist}
    \newenvironment{varproof}%
    {\rm \trivlist \item[\hskip \labelsep{\bf Proof}]}%
    {\hspace*{\fill}$\Box$\endtrivlist}

    \newcommand{\supp}{{\operatorname{supp}}}

    \DeclareMathOperator*{\Tr}{Tr}

\begin{document}
\title{The birth of a cut in unitary random matrix ensembles}
\author{Tom Claeys}

\maketitle

\begin{abstract}
We study unitary random matrix ensembles in the critical regime where a new cut arises away from the original spectrum. We perform a double scaling limit where the size of the matrices tends to infinity, but in such a way that only a bounded number of eigenvalues is expected in the newborn cut.
It turns out that limits of the eigenvalue correlation kernel are given by Hermite kernels corresponding to a finite size Gaussian Unitary Ensemble (GUE). When modifying the double scaling limit slightly, we observe a remarkable transition each time the new cut picks up an additional eigenvalue, leading to a limiting kernel interpolating between GUE-kernels for matrices of size $k$ and size $k+1$. We prove our results using the Riemann-Hilbert approach.
\end{abstract}

\section{Introduction}

We consider unitary invariant random matrix ensembles on the set
of Hermitian $n\times n$ matrices, with a probability density of
the form
\begin{equation}\label{randommatrixensemble}
Z_n^{-1}\exp(-n\Tr V(M))dM,
\end{equation}
where $Z_n$ is a normalization constant and $dM$ is the usual flat
Lebesgue measure on the Hermitian matrices. We assume the
confining potential $V$ to be real analytic on $\mathbb R$ with enough growth at
infinity,
\begin{equation}
    \frac{V(x)}{\log (x^2+1)} \to +\infty \quad\textrm{ as }
    x\to \pm\infty. \label{growthV}
\end{equation}
Eigenvalues of a random matrix in the ensemble
(\ref{randommatrixensemble}) follow a determinantal point process generated by the following
correlation kernel \cite{Mehta},
\begin{equation} \label{kernel}
    K_n(x,y)=
        e^{-\frac{n}{2}V(x)} e^{-\frac{n}{2}V(y)}
        \sum_{k=0}^{n-1} p_k^{(n)}(x) p_k^{(n)}(y),
\end{equation}
given in terms of the orthonormal polynomials
\[
    p_k^{(n)}(x)=\kappa_k^{(n)} x^k + \cdots,
    \qquad\qquad \mbox{$\kappa_k^{(n)}>0$,}
\]
with respect to the weight $e^{-nV}$ on the real line. Using the
Christoffel-Darboux formula, (\ref{kernel}) can also be written in
the following form,
\begin{equation} \label{kernel2}
    K_n(x,y)
        =e^{-\frac{n}{2}V(x)}
e^{-\frac{n}{2}V(y)}\,\frac{\kappa_{n-1}^{(n)}}{\kappa_{n}^{(n)}}\,
\frac{p_n^{(n)}(x)p_{n-1}^{(n)}(y)-p_n^{(n)}(y)p_{n-1}^{(n)}(x)}{x-y}.
\end{equation}

\medskip

If we let the size $n$ of the matrices grow to infinity, the
limiting mean eigenvalue distribution of the ensemble exists and
depends on $V$. In general it can be characterized as the
equilibrium measure $\rho_V$ (see e.g.\ \cite{Deift}) minimizing the
logarithmic energy in external field $V$,
\begin{equation}\label{definition: energy}
    I_{V}(\rho)=
        \iint \log \frac{1}{|x-y|}d\rho(x)d\rho(y)
        +\int V(y)d\rho(y),
\end{equation}
among all probability measures $\rho$ on $\mathbb R$. This
minimization property is equivalent to the following Euler-Lagrange
variational conditions \cite{SaTo}: there exists a constant
$\ell\in\mathbb{R}$ such that
\begin{align}
    \label{variationalcondition:mu-equality}
    & 2\int \log |x-y|d\rho_V(y)-V(x)=\ell,
        &\mbox{for $x\in \supp\,\rho_V$,}
    \\[1ex]
    \label{variationalcondition:mu-inequality}
    & 2\int \log |x-y|d\rho_V(y)-V(x)\leq \ell,
        &\mbox{for $x\in \mathbb R\setminus \supp\,\rho_V$.}
\end{align}
It is known \cite{DKM} that, for real analytic $V$, $\rho_V$ has a
density $\varphi_V$ which can be written in the form
\begin{equation}\label{def Q}
\varphi_V(x)=\frac{1}{\pi}\sqrt{q_V^-(x)},
\end{equation}
where $q_V^-$ denotes the negative part of a real analytic
function $q_V=q_V^+-q_V^-$. Moreover $q_V(\pm x)$ is positive for large real $x$,
from which it readily follows that $\supp\,\rho_V$ is a finite
union of bounded intervals. The endpoints of the support are the
zeros of $q_V$ with odd multiplicity. Generically the following
conditions hold \cite{KM}:
\begin{itemize}
\item[(a)] the variational inequality
(\ref{variationalcondition:mu-inequality}) holds strictly for
$x\in\mathbb R\setminus \supp\,\rho_V$, \item[(b)] $\varphi_V$ is positive in the interior of its support, or equivalently, $q_V^-$ has no
zeros in the interior of $\supp\,\rho_V$, \item[(c)] $\varphi_V$ behaves like a square root near the endpoints of its support, or equivalently, the zeros of
$q_V$ with odd multiplicity are simple.
\end{itemize}
In the critical cases where the above generic conditions do not
hold, singular points occur. According to \cite{DKMVZ2}, singular
points are classified as follows.
\begin{itemize}
\item[(i)] Type I singular points or singular exterior points:
these are isolated points outside $\supp\,\rho_V$ where equality
in (\ref{variationalcondition:mu-inequality}) holds. Here $q_V$
vanishes at an order $4m-2$ for $m\geq 1$. \item[(ii)] Type II
singular points or singular interior points: these are points in
the interior of $\supp\,\rho_V$ where the density $\varphi_V$
vanishes. Necessarily $q_V$ has a zero of multiplicity $4m$ at
such a point. \item[(iii)] Type III singular points or singular
edge points: these are endpoints of $\supp\,\rho_V$ where the
density $\varphi_V$ vanishes faster than a square root. Here the
only possibilities are that $q_V$ has a zero of multiplicity
$4m+1$ for $m\geq 1$.
\end{itemize}

When varying the potential $V$, the critical ensembles are the ones
where a change in the number of intervals of $\supp\,\rho_V$ may
possibly occur. Type I singular points correspond with the birth of
a new cut away from the spectrum. Type II singular points indicate
the closing of a gap in between two intervals of the support. Near
type III singular points, a new interval can arise at the edge of
the spectrum, or in other words, a gap can close simultaneously with
one of the cuts.

\begin{remark}\label{remark 0}
It is important to note that not all multiple zeros of $q_V$ are
singular points. For example, if $\supp\,\rho_V$ consists of one
interval $[a,b]$, it is known \cite{Deift} that
\begin{equation}\label{varQ}
2\int \log |x-y|d\rho_V(y)-V(x)-\ell=-2\int_b^x q_V^{1/2}(y)dy
\qquad \mbox{ for $x>b$.}
\end{equation}
Consequently, in order to have a type I singular (exterior) point
$x^*>b$, $q_V^{1/2}$ should necessarily change sign at some point
in between $b$ and $x^*$. At this intermediate point, $q_V$ has a
multiple zero although in general it is not a singular point.
\end{remark}

The local behavior of the eigenvalues of large random matrices near
some reference point $x^*$ depends on the 'nature' of $x^*$. Here
'nature' refers to the behavior of the limiting mean eigenvalue
density near $x^*$. The two regular kinds of points that occur, are
points in the bulk of the spectrum (where $\varphi_V$ is positive)
and points at the edge of the spectrum where $\varphi_V$ vanishes
like a square root. The critical ensembles give lead, as described
above, to three additional types of points that correspond to
different local behavior of the eigenvalues.

Local scaling limits of the eigenvalue correlation kernel
(\ref{kernel}) turn out to be universal, which means that they
depend on the nature of  the reference point $x^*$, but not on the
confining potential $V$ nor on the position of $x^*$. In the bulk of
the spectrum this leads one to the sine kernel \cite{BI1, Deift, DKMVZ2, DKMVZ1, PS},
\[\lim_{n\to\infty}\frac{1}{\pi\varphi_V(x^*)n}K_n(x^*+\frac{u}{\pi\varphi_V(x^*)n},
x^*+\frac{v}{\pi\varphi_V(x^*)n})=\frac{\sin\pi(u-v)}{\pi(u-v)}.
\]
Near a regular edge point, the limiting kernel is given in terms
of Airy functions \cite{DG2},
\[\lim_{n\to\infty}\frac{1}{c_Vn^{2/3}}K_n(x^*+\frac{u}{c_Vn^{2/3}},
x^*+\frac{v}{c_Vn^{2/3}})=\frac{\Ai(u)\Ai'(v)-\Ai(v)\Ai'(u)}{u-v},
\] for some constant $c_V$.

\medskip

As already mentioned, singular points indicate a transition where
the number of intervals in the spectrum can change. These transitions
can be observed when including a parameter in the confining potential,
$V=V_t$. If the singular point corresponds to the value of $t=t_c$,
double scaling limits of the kernel, where we let $n\to\infty$ and
at the same time we let $t\to t_c$ at an appropriate rate, lead to
families of limiting kernels, depending on a parameter. Near singular interior
points these double scaling limits are given by a kernel related to
the Hastings-McLeod solution of the Painlev\'e II equation, see
\cite{BI2, CK, CKV, Shcherbina}. In the singular edge case, the limiting
kernels are related to a special solution of an equation in the
Painlev\'e I hierarchy \cite{BB, CV2}. Singular exterior points where a new cut is born, have
been studied in \cite{Eynard}, but rigorous results about the
limiting eigenvalue correlation kernel are not available in the
literature yet. It is the aim of this paper to obtain rigorous asymptotics for the correlation kernel near a singular exterior point in the birth of a new cut. Unlike in the two other critical cases, there are no Painlev\'e equations involved describing the local behavior of eigenvalues. In this case, the limit of the eigenvalue correlation kernel will be given by a kernel corresponding to a finite size Gaussian Unitary Ensemble (GUE). The size of the relevant GUE will depend on the precise choice of double scaling limit we take, or on the number of eigenvalues that are expected in the new cut.

\medskip

Similar transitions as the ones described above also occur in the study of the small dispersion limit of the Korteweg-De Vries equation \cite{GK}, and in a more general context when studying Hamiltonian perturbations of hyperbolic systems \cite{DubrovinI, Dubrovin}. In transitional regimes where algebraic asymptotics for a solution of the KdV equation turn into elliptic asymptotics, phenomena are observed which are expected to correspond to the transitions for unitary random matrix ensembles, corresponding to the three types of singular points. Painlev\'e asymptotics similar to those near singular interior points and singular edge points in random matrix ensembles have been verified numerically \cite{GK}, but not much is known about singular exterior points. This is yet another motivation to study the birth of a cut.

\subsection{Statement of results}

The aim of this paper is to obtain a double scaling limit of the
eigenvalue correlation kernel near singular exterior points. We
deal with the case where $q_V$ has a double zero at $x^*$, which
is the lowest possible order of vanishing for a type I singular
point.

\medskip

We consider a one-parameter family of potentials $V_t=V/t$, where
$V=V_1$ is such that a singular exterior point $x^*$ is present, and such that $\supp\,\rho_V=[a,b]$. Furthermore we assume that there are, besides $x^*$, no
other singular points. For such a potential $V$, it was shown in \cite{KM} that
$\rho_t:=\rho_{V_t}$ is supported on one interval $[a_t,b_t]$
for $t$ slightly less than $1$, without singular points.
It was also shown in this paper that for $t$ slightly bigger than $1$, a new
cut is born near $x^*$, so that
$\supp\,\rho_t=[a_t,b_t]\cup[\alpha_t,\beta_t]$. As
$t\searrow 1$, the cut disappears, so that $\alpha_t,\beta_t\to x^*$.
The restrictions that $\supp\,\rho_V$ is supported on one interval, and that there are no other singular points besides $x^*$, are technical rather than crucial. We expect that universality remains valid for potentials $V$ that do not satisfy those assumptions.

\medskip

We work in a double scaling regime where we let the size $n$ of
the matrices tend to infinity, and at the same time we let $t\to
1$ in such a way that \[|t-1|\leq M\frac{\log n}{n}\] for some
$M>0$ which can be arbitrary large. This double scaling limit
implies that a bounded number of eigenvalues is expected in the
vicinity of $x^*$.

\medskip

Let us now formulate the main result of the present work.

\begin{theorem}\label{theorem: universality}
Let $V$ be real analytic satisfying condition (\ref{growthV}), and
assume that $x^*$ is a type I singular point where $q_V$ has a
double zero. Assume also that $\supp\,\rho_V=[a,b]$ with $b<x^*$, and that there are no other singular points
besides $x^*$. Let $K_{n,t}$ be the eigenvalue correlation kernel
(\ref{kernel}) for the potential $V_t=V/t$. We take a double
scaling limit where $n\to\infty$ and $t\to 1$ in such a way that
$|t-1|\leq M\frac{\log n}{n}$. We define
\begin{equation}s:=
2(t-1)\frac{n}{\log n}\ {\small\int_b^{x^*}\hspace{-0.4cm}\frac{1}{\sqrt{(s-a)(s-b)}}\,ds},\end{equation} so
that $s$ remains bounded in the double scaling limit.
 Then, depending on the value of $s$, we have the following limits for the eigenvalue correlation kernel,
\begin{multline}\label{kernel1}
\lim\frac{1}{(cn)^{1/2}}K_{n,t}\left(x^*+\frac{u}{(cn)^{1/2}},
x^*+\frac{v}{(cn)^{1/2}}\right)\\
=\begin{cases}\begin{array}{ll}\mathbb K^{\rm GUE}(u,v;k)&\mbox{
for $k-\frac{1}{2}<s<k+\frac{1}{2}$, $k\geq 1$}, \\[1.5ex]
0&\mbox{ for $s<\frac{1}{2}$,}
\end{array}
\end{cases}
\end{multline}
with $c$ given by
\begin{equation}\label{c}
c=\frac{1}{\sqrt 2}q_V''(x^*)^{1/2},
\end{equation}
and $\mathbb K^{\rm GUE}$ is given by
\begin{equation}\label{limitingkernel}
\mathbb K^{\rm
GUE}(u,v;k)=\sqrt\frac{k}{2}\ e^{-\frac{u^2+v^2}{2}}\ \frac{H_{k}(u)H_{k-1}(v)-H_{k}(v)H_{k-1}(u)}{u-v},\qquad\mbox{
as $k\geq 1$},
\end{equation}
where we write $H_k$ for the $k$-th degree normalized Hermite polynomial, with leading coefficient $\frac{2^{k/2}}{\pi^{1/4}\sqrt{k!}}$,
orthonormal with respect to the weight $e^{-x^2}$.
\end{theorem}

\begin{remark}
Obviously, the limits in (\ref{kernel1}) do not hold not uniformly for $s$ near a half positive integer. They do hold uniformly for $s$ bounded and away from arbitrary small fixed neighborhoods of the half positive integers.
\end{remark}

\begin{remark}
$\mathbb K^{\rm GUE}(.,.;k)$ is the eigenvalue correlation kernel
for the $k\times k$ GUE, which is the random matrix ensemble (\ref{randommatrixensemble}) for the potential $V(x)=x^2$. After re-scaling, the eigenvalues in the newborn cut seem to behave
asymptotically in the same way as the eigenvalues in a finite GUE.
\end{remark}

\begin{remark}
For $s<1/2$, the limiting eigenvalue correlation kernel is trivial. This is
not surprising since no eigenvalues are expected in the vicinity
of $x^*$ for $t<1$. A first eigenvalue near $x^*$ is only expected when $s$ approaches $1/2$. Each time we shift $s$ with $1$, an
additional eigenvalue is expected in the new cut.
\end{remark}

\medskip

Theorem \ref{theorem: universality} gives us the limiting
eigenvalue correlation kernel in the case where $s$ is not a half
positive integer. It is natural to ask what happens when $s$ is
close to a half positive integer. It seems that, if $s$ increases and passes a half positive integer, an additional eigenvalue is picked up by the new cut with high probability, which leads to a kernel corresponding to a GUE of larger size. Near the half integers, a remarkable transition takes place, involving a limiting kernel interpolating between a GUE kernel for matrices of size $k$ and $k+1$.

\medskip

\begin{theorem}\label{theorem: universality 2}
Under the same conditions as in Theorem \ref{theorem:
universality}, there exist sequences $\lambda_{n,t}^\pm$ such that the following asymptotics hold in the
double scaling limit,
\begin{multline}\label{asymptotic expansion kernel}
\frac{1}{(cn)^{1/2}}K_{n,t}\left(x^*+\frac{u}{(cn)^{1/2}},x^*+\frac{v}{(cn)^{1/2}}\right)\\
=\begin{cases}\begin{array}{ll}\lambda_{n,t}^-\mathbb K^{\rm
GUE}(u,v;k)& \\
\qquad\qquad+\lambda_{n,t}^+\mathbb K^{\rm
GUE}(u,v;k+1)+\bigO\left(\frac{\log n}{n^{1/2}}\right),&\mbox{
for $k\leq s \leq k+1$, $k\geq 0$,}\\[1.5ex]
\bigO(n^{-1/2}),&\mbox{ for $s<0$.}
\end{array}
\end{cases}
\end{multline}
Furthermore the sequences $\lambda_{n,t}^\pm$ are such that
\[\lambda_{n,t}^+ + \lambda_{n,t}^-=1,\] and
\begin{align}
\label{lambda1}&\lambda_{n,t}^+=1-\lambda_{n,t}^-=\bigO(n^{-1/2+s-k}), &\mbox{ as $k\leq s \leq k+1/2,$}\\
\label{lambda2}&\lambda_{n,t}^-=1-\lambda_{n,t}^+=\bigO(n^{1/2+k-s}), &\mbox{ as $k+1/2\leq s \leq k+1.$}
\end{align}
For simplicity in notation, we have written $K^{\rm GUE}(u,v;0)=0$ in (\ref{asymptotic expansion kernel}). The expansion (\ref{asymptotic expansion kernel}) holds uniformly for $s$ bounded.
\end{theorem}
\begin{remark}
Note that (\ref{asymptotic expansion kernel}) is compatible with (\ref{kernel1}). Indeed one observes that, for $s$ away from a half positive integer, either $\lambda_{n,t}^+$ of $\lambda_{n,t}^-$ tends to $0$, so that, in the limit, we are only left with one of the two kernels $\mathbb K^{\rm
GUE}(u,v;k)$ and $\mathbb K^{\rm
GUE}(u,v;k+1)$. This is exactly what is stated in Theorem \ref{theorem: universality}.
Unfortunately, we are not able to give simple formulas for the sequences $\lambda_{n,t}^\pm$. We can give formulas for them, as we will do in Section \ref{section proof}, but we have been unable to reduce those formulas to simple expressions.
\end{remark}
\begin{remark}
For $s$ near a half positive integer, the limiting kernel changes abruptly. If we would put
\[s=k+\frac{1}{2}+\frac{\xi}{\log n},\]
we expect that this would lead us to a limiting kernel of the form
\[(1-\lambda)\mathbb K^{\rm
GUE}(u,v;k)+\lambda\mathbb K^{\rm
GUE}(u,v;k+1),\] with $0 \leq \lambda \leq 1$.
This is a kernel which corresponds to a determinantal point process as well, just like the GUE-kernels. The relevant point process is the one where we have $k$ GUE-eigenvalues with probability $1-\lambda$ , and $k+1$ GUE-eigenvalues with probability $\lambda$.
\end{remark}

\begin{example}
A concrete example of a random matrix ensemble where a singular edge
point is present, was given in \cite{Eynard}. For the potential
\begin{equation}
V(x)=\frac{1}{1+e\tilde e}\left(\frac{1}{4}x^4-\frac{e+\tilde
e}{3}x^3+\frac{e\tilde e -2}{2}x^2+2(e+\tilde e)x\right),
\end{equation}
with $e>2$ and $\tilde e$ such that $\int_2^e(x-e)(x-\tilde
e)\sqrt{x^2-4}\,dx=0$, there is a singular exterior point at
$x^*=e$, where $q_V$ has a double zero. For a potential of degree
less then $4$, singular points cannot occur. In order to construct
an example for which there is a singular exterior point where
$q_V$ has a zero of order $4m-2$, a potential $V$ of
degree at least $2m+2$ is needed. In those higher order cases, we expect that double scaling limits can be tuned in such a way that the limiting eigenvalue correlation kernel is no longer related to Hermite polynomials, but to polynomials orthogonal with respect to a weight of the form $e^{-P(x)}$, where $P$ can be any polynomial of degree $2m$.
\end{example}

\subsection{Outline for the rest of the paper}

In Section \ref{section equilibrium}, we will construct equilibrium measures which correspond to the potential $V_t$. For technical reasons, we need modified measures compared to the usual equilibrium measures used in e.g.\ \cite{Deift, DKMVZ2, DKMVZ1}, and also different from the modified measures used in \cite{CK, CV2}.
In Section \ref{section RH}, we recall the Riemann-Hilbert (RH) problem for orthogonal polynomials introduced by Fokas, Its, and Kitaev \cite{FokasItsKitaev}. We follow the ideas of the Deift/Zhou steepest descent method \cite{DZ1} in order to find asymptotics for the orthogonal polynomials. Here we follow similar lines as in \cite{Deift, DKMVZ2, DKMVZ1}, with however two major differences. The first one is, as already mentioned, the use of modified equilibrium measures, and the second one is the construction of a local parametrix near the singular point $x^*$. For this construction, we will need, in Section \ref{section: local}, a model RH problem built out of Hermite polynomials.
The local parametrix will enable us to find asymptotics for the orthogonal polynomials near $x^*$. In Section \ref{section proof}, we will use those asymptotics to obtain asymptotics for the eigenvalue correlation kernel and to prove Theorem \ref{theorem: universality} and Theorem \ref{theorem: universality 2}.

\section{Equilibrium measures}\label{section equilibrium}
It is a well-known fact that a $g$-function related to an equilibrium measures plays a crucial
role in the Deift/Zhou steepest descent analysis. As already mentioned in the introduction, the limiting mean eigenvalue
distribution $\rho_t$ is an equilibrium measure in external field $V_t$, and would hence be an obvious candidate to built out the $g$-function.
However this would not be a convenient choice because the endpoints $\alpha_t$ and
$\beta_t$ of the new cut vary with $t$ and both tend to $x^*$ as
$t\searrow 1$. This would create various technical difficulties for the construction of a local parametrix near the critical point $x^*$, which will be the most crucial issue in our RH analysis.

To prevent the presence of endpoints near $x^*$ varying with $t$, we will construct for $t>1$ a modified
equilibrium measure $\mu_{t}$ which we force to have its support
away from $x^*$. The portion of mass of $\rho_t$ near $x^*$ can
however not be ignored, and for this purpose we add to $\mu_t$ a point mass centered at some point close to $x^*$.
The rough idea is that we can 'approximate' the limiting mean
eigenvalue density $\rho_t$ by a probability measure
of the form $\widehat\mu_{n,t}=\mu_{n,t} + m_{n,t}\delta_{x_{n,t}^*}$, where
\begin{itemize}
\item $\mu_{n,t}$ is the positive equilibrium measure in external field $V_t$ with mass $1-m_{n,t}$, with $m_{n,t}$ given by
\begin{equation}\label{definition mnt}
m_{n,t}=\max\left\{\frac{s}{n},0\right\}=\begin{cases}\begin{array}{ll}
2\frac{t-1}{\log n}\ \int_b^{x^*}\hspace{-0.2cm}\frac{1}{\sqrt{(s-a)(s-b)}}ds&\mbox{ as $t>1$},\\
0 &\mbox{ as $t\leq 1$.}\end{array}\end{cases}
\end{equation}
\item $\delta_{x_{n,t}^*}$ is the Dirac distribution centered with mass $1$ at
$x_{n,t}^*$,
\item $x_{n,t}^*$ is a point near $x^*$ which we will
determine below.
\end{itemize}

To be precise in defining $\mu_{n,t}$, it is the equilibrium measure minimizing
the logarithmic energy $I_{V_t}(\mu)$, defined by
(\ref{definition: energy}), among all positive measures $\mu$ satisfying
the following two conditions: \begin{itemize}\item
$\supp\,\mu\subset \mathbb R\setminus[x^*-\epsilon, x^*+\epsilon]$
for some sufficiently small fixed $\epsilon>0$, \item $\mu(\mathbb
R)=1-m_{n,t}$, with $m_{n,t}$ defined by (\ref{definition mnt}).
\end{itemize}
The equilibrium measure $\mu_{n,t}$ does not depend on the choice of
$\epsilon$, at least not if $\epsilon$ is sufficiently small so that $[x^*-\epsilon, x^*+\epsilon]$ does not intersect with $[a,b]$, and if $t$ is sufficiently close to $1$.
Note also that, for $t\leq 1$, our new equilibrium measure is exactly equal to $\rho_t$. The measure $\mu_{n,t}$ is thus independent of $n$ for $t\leq 1$.

\medskip

Since there are no singular points in external field $V$ any longer after the exclusion of $[x^*-\epsilon, x^*+\epsilon]$, it was shown in \cite{KM} that for $t$ sufficiently close to $1$, $\supp\,\mu_{n,t}:=[a_{n,t}',b_{n,t}']$ consists of one single interval.
It is a standard fact \cite{SaTo} that $\mu_{n,t}$ satisfies the following variational conditions,
with $V_t=V/t$, for $t$ sufficiently close to $1$,
\begin{align}
    \label{variationalcondition:hatmu-equality}
    &2\int \log |x-y|d\mu_{n,t}(y)-V_t(x)=\ell_{n,t},
        \qquad\mbox{for $x\in [a_{n,t}',b_{n,t}']$,}\\
    \label{variationalcondition:hatmu-inequality}
     &2\int \log |x-y|d\mu_{n,t}(y)-V_t(x)< \ell_{n,t},
       \qquad\mbox{for $x\in \mathbb R\setminus ([a_{n,t}',b_{n,t}']\cup
       [x^*-\epsilon, x^*+\epsilon])$.}
\end{align}
Those variational conditions will be crucial throughout the following sections.

\medskip

We can directly apply the results obtained in \cite{DKM} to conclude that the density $\psi_{n,t}$ of
$\mu_{n,t}$ can be written in the form
\begin{equation}\label{definition hatQ}
\psi_{n,t}(x)=\frac{1}{\pi}\sqrt{Q_{n,t}^-(x)},
\end{equation}
where $Q_{n,t}$ is a real analytic function with negative part
$Q_{n,t}^-$. Also it follows from \cite{DKM}
that
\begin{equation}\label{DKM identity q}
Q_{n,t}(z)=\left(\frac{V'(z)}{2t}\right)^2-\frac{1}{t}\int
\frac{V'(z)-V'(y)}{z-y}d\mu_{n,t}(y).
\end{equation}
Using this identity, weak* convergence of $\mu_{n,t}$ to $\rho_1=\mu_{n,1}$
is enough to conclude that $Q_{n,t}(z)\to q_V(z)$ uniformly on
compact sets in a neighborhood of the real line, as $t\to 1$ and $n\to \infty$. This means that
$Q_{n,t}$ has simple zeros $a_{n,t}'$ and $b_{n,t}'$ tending to
$a$ and $b$ as $t\to 1$, $n\to\infty$. However we can do more. It follows from a result by Buyarov and Rahmanov \cite{BR} that
\begin{equation}
\mu_{n,t}-\mu_{n,1}=\bigO(t-1),
\end{equation}
Using (\ref{DKM identity q}) we now have that
$Q_{n,t}(x)-q_V(x)=\bigO(t-1)$ uniformly on compact sets, which
implies that $a_{n,t}'=a+\bigO(t-1)$ and $b_{n,t}'=b+\bigO(t-1)$.

\medskip

It is convenient to rewrite (\ref{definition hatQ}) in the following way,
\begin{equation}\label{hatrho-h}
\psi_{n,t}(x)=\frac{1}{\pi}\sqrt{(b_{n,t}'-x)(x-a_{n,t}')}\
h_{n,t}(x)\chi_{[a_{n,t}', b_{n,t}']}(x),
\end{equation}
where $h_{n,t}$ is real analytic on $\mathbb R$ and
$h_{n,t}=h_{n,1}+\bigO(t-1)$ uniformly on compact sets as $t\to
1$, $n\to\infty$. Now it follows from formula $(\ref{def Q})$ and
the fact that $q_V$ has a double zero at $x^*$, that $h_{n,1}$ has
a simple zero at $x^*$. It then follows that $h_{n,t}$ must have a
zero $x_{n,t}^*=x^*+\bigO(t-1)$ as $t\to 1$, $n\to\infty$. This
zero $x_{n,t}^*$ is the point where we center the Dirac measure
approximating the portion of the limiting mean eigenvalue distribution in the new cut. Note that the point
$x_{n,t}^*$ is a double zero of $Q_{n,t}$, but it is not a singular
point, cf.\ Remark \ref{remark 0}.

\medskip

So far, we did not give any arguments why we need to choose $m_{n,t}$ and $x_{n,t}^*$ in the way we did. Only at the very end of the RH analysis, in the construction of a local parametrix near $x^*$ in Section \ref{section: local}, it will become clear that only these choices for $m_{n,t}$ and $x_{n,t}^*$ do the job. Except for the construction of the local parametrix, the RH analysis would also work for other values for $m_{n,t}$ and $x_{n,t}^*$.

\section{Riemann-Hilbert analysis}\label{section RH}
The starting point of our analysis is the RH problem for
orthogonal polynomials introduced by Fokas, Its, and Kitaev
\cite{FokasItsKitaev}. We follow the approach of \cite{Deift,
DKMVZ2, DKMVZ1}, where the Deift/Zhou steepest descent method
\cite{DZ1} was used in order to find large $n$ asymptotics for the
solution of this RH problem. The main idea is to apply a series of
transformations to the RH problem in order to find, at the end, a
RH problem which can be solved approximately for large $n$. The
crucial new feature in our analysis is, besides the modification
of the equilibrium measure described in the previous section, the
construction of a local parametrix near the singular exterior
point $x^*$ in a double scaling limit. This will be done in
Section \ref{section: local} using the RH problem for Hermite polynomials.
The remaining part of the Deift/Zhou
steepest descent analysis follows similar lines as the ones
developed in \cite{Deift, DKMVZ2, DKMVZ1} and later also applied
to double scaling limits in \cite{CK, CKV, CV2}.

\medskip

For technical reasons, we assume that the confining potential $V$ is such
that the support of the limiting mean eigenvalue density consists of one single interval, $\supp\,\rho_V=[a,b]$.
We also assume that the singular point lies at the right side of the spectrum, $x^*>b$, which we can do without loss of
generality because of the possibility to consider the potential
$V(-x)$ instead of $V(x)$.
Another technical restriction is that we assume the absence of any other singular point besides $x^*$.

\subsection{RH problem for orthogonal polynomials}

For each $n$ and $t$, we consider the following RH problem. We seek
for a $2\times 2$ matrix-valued function $Y(z)=Y(z;n,t)$ satisfying
the following conditions.

\subsubsection*{RH problem for $Y$:}

\begin{itemize}
    \item[(a)] $Y:\mathbb{C}\setminus \mathbb{R}\to\mathbb{C}^{2\times 2}$ is analytic,
    \item[(b)] $Y$ has continuous boundary values $Y_\pm(x)$
        for $x\in\mathbb{R}$, where $Y_+(x)$ and $Y_-(x)$ denote the limiting values when
        approaching $x$ from above and below, and
        \begin{equation}\label{RHP Y: b}
            Y_+(x)=Y_-(x)
                \begin{pmatrix}
                    1 & e^{-nV_{t}(x)} \\
                    0 & 1
                \end{pmatrix},
                \qquad  \mbox{for $x \in\mathbb{R}$.}
        \end{equation}
    \item[(c)] $Y$ has the following asymptotic behavior at
    infinity,
        \begin{equation}\label{RHP Y: c}
            Y(z)=\left(I+\bigO(z^{-1})\right)
                \begin{pmatrix}
                    z^n & 0 \\
                    0 & z^{-n}
                \end{pmatrix},
                \qquad  \mbox{as $z\rightarrow \infty$.}
        \end{equation}
\end{itemize}

\medskip

This RH problem has a unique solution which is given in terms of
the orthonormal polynomials $p_k=p_k^{(n,t)}$ with respect to the
weight $e^{-nV_t}$ on $\mathbb R$, see \cite{FokasItsKitaev},
\begin{equation}\label{RHP Y: solution}
    Y(z)=
    \begin{pmatrix}
        \kappa_n^{-1}p_n(z) &
            \displaystyle{\frac{\kappa_n^{-1}}{2\pi i}\int_{\mathbb R}\frac{p_n(u)
            e^{- n V_{t}(u)}}{u-z}\,
            du}
        \\[3ex]
        - 2\pi i \kappa_{n-1} p_{n-1}(z) &
            \displaystyle{-\kappa_{n-1} \int_{\mathbb R}\frac{p_{n-1}(u) e^{- n
            V_{t}(u)}}{u-z}\,du}
    \end{pmatrix},\qquad\mbox{for $z\in\mathbb C \setminus\mathbb R$.}
\end{equation}
Here we have written $\kappa_k=\kappa_k^{(n,t)}>0$ for the leading
coefficient of $p_k$.

It is straightforward to check that the eigenvalue correlation
kernel $K_{n,t}$, given by (\ref{kernel2}), can be expressed in
terms of $Y$. Using the fact that $\det Y \equiv 1$ (this follows
from the RH conditions for $Y$ using a standard complex analysis
argument), one obtains
\begin{equation}\label{KinY}
    K_{n,t}(x,y) =
        e^{-\frac{n}{2}V_{t}(x)}e^{-\frac{n}{2}V_{t}(y)} \frac{1}{2\pi i(x-y)}
        \begin{pmatrix}
            0 & 1
        \end{pmatrix}
        Y_{\pm}^{-1}(y) Y_{\pm}(x)
        \begin{pmatrix}
            1 \\ 0
        \end{pmatrix}.
\end{equation}

From (\ref{KinY}), it is clear that all the information needed to
prove Theorem \ref{theorem: universality} is contained in the RH
solution $Y$. We now need to find sufficiently accurate
asymptotics for $Y$ in the double scaling limit.

\subsection{First transformation $Y\mapsto T$}\label{section: T}

In this section we perform a transformation $Y\mapsto T$ of the RH
problem which normalizes the behavior of the RH solution at infinity in such a way that $T(z)\to I$ as $z\to\infty$. Besides that, the
map will modify the jumps in a convenient way. A crucial role in
the transformation is played by the so-called $g$-function, which is related to the equilibrium measures constructed in Section \ref{section equilibrium}.

The standard way to define the $g$-function would be to put it equal
to $\int\log(z-s)d\rho_{t}(x)$, where $\rho_t$ is the limiting
eigenvalue distribution of the random matrix ensemble. As already
noted in the previous section, this would involve several
technical difficulties. We have already anticipated to these
difficulties by replacing the limiting mean eigenvalue
distribution $\rho_t$ with a modified measure
$\widehat\mu_{n,t}=\mu_{n,t}+m_{n,t}\delta_{x_{n,t}^*}$, where we
approximated the portion of mass in the new cut by a Dirac
measure centered at $x_{n,t}^*$. The new equilibrium measure leads us to the
following $g$-function,
\begin{equation}\label{definition: g}
g(z)=g_{n,t}(z)=\int\log(z-y)d\widehat\mu_{n,t}(y)=\int\log(z-y)d\mu_{n,t}(y)
+ m_{n,t} \log(z-x_{n,t}^*),
\end{equation}
where we take $\log z$ analytic in $\mathbb C\setminus
(-\infty,0]$ with $-\pi <\Im\log z < \pi$.
The variational conditions (\ref{variationalcondition:hatmu-equality})
and (\ref{variationalcondition:hatmu-inequality}) can now be translated to properties for the $g$-functions. For $t$ sufficiently close to $1$, we have that, with $\ell=\ell_{n,t}$,
\begin{align}
    &\label{property g: 1}g_{+}(x)+g_{-}(x)-V_{t}(x)-\ell=2m_{n,t}\log|x_{n,t}^*-x|
    ,\qquad\qquad\mbox{for
    $x\in [a_{n,t}',b_{n,t}']$,}\\
    &g_{+}(x)+g_{-}(x)-V_{t}(x)-\ell-2m_{n,t}\log|x_{n,t}^*-x|<0,\nonumber\\
    &\label{property g: 2}\hspace{5.87cm}\mbox{for
    $x\in \mathbb R\setminus([a_{n,t}', b_{n,t}']
    \cup [x^*-\epsilon, x^*+\epsilon])$.}
\end{align}
Note also that
\begin{equation}\label{property g: 4}
    g_{+}(x)-g_{-}(x)=2\pi i\int_x^{+\infty} d\widehat\mu_{n,t}(y),\qquad
    \mbox{for $x\in\mathbb
    R$,}
\end{equation}
which means in particular that
\begin{equation}\label{property g: 3}
    g_{+}(x)-g_{-}(x)=
    \begin{cases}
        2\pi i, & \mbox{for $x<a_{n,t}'$,} \\
        2\pi im_{n,t}, & \mbox{for $b_{n,t}'<x<x_{n,t}^*$,}\\
        0, & \mbox{for $x>x_{n,t}^*$.}
    \end{cases}
\end{equation}

\medskip

The above properties of the $g$-function are, together with the fact that $e^{ng(z)}=z^n(1+\bigO(1/z))$ as $z\to\infty$, crucial to transform
the RH problem. We define $T$ as follows,
\begin{equation}\label{def T}
    T(z)=e^{-\frac{n}{2}\ell\sigma_3} Y(z) e^{-ng(z)\sigma_3}
    e^{\frac{n}{2}\ell\sigma_3}, \qquad \mbox{for
    $z\in\mathbb{C}\setminus\mathbb R$,}
\end{equation}
where $\sigma_3$ is the third Pauli matrix,
$\sigma_3=\begin{pmatrix}1&0\\0&-1\end{pmatrix}$. With
$\nu=\max\{s,0\}=nm_{n,t},$ it is straightforward to check using (\ref{property g: 1}), (\ref{property g: 2}), and (\ref{property g: 3}), that $T$
satisfies the following RH problem.

\subsubsection*{RH problem for $T$:}
\begin{itemize}
    \item[(a)] $T:\mathbb{C}\setminus \mathbb{R}\to\mathbb{C}^{2\times 2}$ is analytic,
    \item[(b)] $T_+(x)=T_-(x)v_T(x)$ for $x\in\mathbb{R}$, with
        \begin{equation} \label{RHP T: b}
            v_T(x)=
            \begin{cases}
                \begin{pmatrix}
                    1 & e^{n(g_{+}(x)+g_{-}(x)-V_{t}(x)- \ell)} \\
                    0 & 1
                \end{pmatrix}, & \mbox{for $x\in (-\infty, a_{n,t}')\cup
                (x_{n,t}^*,+\infty)$,}\\[3ex]
                \begin{pmatrix}
                    e^{-n(g_{+}(x)-g_{-}(x))} & |x-x_{n,t}^*|^{2\nu} \\
                    0 & e^{n(g_{+}(x)-g_{-}(x))}
                \end{pmatrix}, & \mbox{for $x\in [a_{n,t}',b_{n,t}']$,} \\[3ex]
                \begin{pmatrix}
                    e^{-2\pi i\nu} & e^{n(g_{+}(x)+g_{-}(x)-V_{t}(x)- \ell)} \\
                    0 & e^{2\pi i\nu}
                \end{pmatrix}, & \mbox{for $x\in (b_{n,t}', x_{n,t}^*)$,}
            \end{cases}
        \end{equation}
    \item[(c)] $T(z)=I+\bigO(1/z)$,\qquad as $z\to\infty$,
    \item[(d)] $T(z)(z-x_{n,t}^*)^{\nu\sigma_3}$ is bounded near $x^*$.

\end{itemize}

It is necessary to add condition (d), controlling the behavior of $T$ near $x^*$, in order to have unique RH solution. A similar condition was not stated in the RH problem for $Y$, where we assumed continuous boundary values and thus a bounded RH solution near $x^*$. In exception of $x_{n,t}^*$, $T$ has continuous boundary values on $\mathbb R$ as well.

\medskip

Using (\ref{KinY}) and (\ref{def T}), we find the following identity for the eigenvalue
correlation kernel in terms of the new RH solution $T$,
\begin{equation}\label{KinT}
    K_{n,t}(x,y) =
        e^{-\frac{n}{2}V_{t}(x)}e^{-\frac{n}{2}V_{t}(y)}e^{ng_{+}(x)}
        e^{ng_{+}(y)}e^{-n\ell} \frac{1}{2\pi i(x-y)}
        \begin{pmatrix}
            0 & 1
        \end{pmatrix}
        T_{+}^{-1}(y) T_{+}(x)
        \begin{pmatrix}
            1 \\ 0
        \end{pmatrix}.
\end{equation}

\subsection{Second transformation $T\mapsto S$}\label{section T}

Before going on with the next transformation of the RH problem, we first write the jump matrices $v_T$ in a slightly
different and more convenient way. Let us define a
function $\phi=\phi_{n,t}$, analytic in $\mathbb C\setminus [-\infty, b_{n,t}']$, by
\begin{equation}
\label{def phi} \phi(z)= \int_z^{b_{n,t}'}Q_{n,t}^{1/2}(s)ds=\int_z^{b_{n,t}'}(s-a_{n,t}')^{1/2}(s-b_{n,t}')^{1/2}h_{n,t}(s)ds,
\end{equation}
with $Q_{n,t}$ as in (\ref{definition hatQ}) and $h_{n,t}$ as in (\ref{hatrho-h}), and with a suitable branch of the square root such that
\[\begin{array}{ll}
\Re \phi_\pm(x)<0, & \mbox{ as $x<a_{n,t}'$,}\\[1.5ex]
\pm\Im \phi_\pm(x)>0, & \mbox{ as $x\in[a_{n,t}',
b_{n,t}']$,}\\[1.5ex]
\phi(x)<0, & \mbox{ as $x\in (b_{n,t}',+\infty)\setminus[x^*-\delta,x^*+\delta]$.}
\end{array}\]
Using the definition of $\phi$, (\ref{hatrho-h}), and (\ref{property g: 3}), we observe that
\begin{equation}\label{g phi}g_+(x)-g_-(x)=\pm 2\phi_\pm(x)+2\pi im_{n,t}, \qquad\mbox{ for $x\in[a_{n,t}',b_{n,t}']$.}\end{equation}
On the other hand, by (\ref{property g: 1}), we have that
\[g_\pm(x)-g_\mp(x)=2g_\pm(x)- V_t(x)-\ell- 2m_{n,t}\log |x_{n,t}^*-x|, \qquad\mbox{ for $x\in[a_{n,t}',b_{n,t}']$.}\]
Combining the two above equations and using the identity theorem, we obtain the following identity,
\begin{equation}\label{equation g phi}
2\phi(z)+2m_{n,t}\log(z-x_{n,t}^*)=2g(z)-V_t(z)-\ell,\qquad \mbox{ for $z\in\mathbb C\setminus (-\infty,x_{n,t}^*)$.}
\end{equation}

\medskip

We can now rewrite the jump matrix $v_T$ given by (\ref{RHP T: b}), using (\ref{g phi}) and (\ref{equation g phi}). For $t>1$ we have the following jump matrix,
\begin{equation} \label{RHP T: b2}
            v_T(x)=
            \begin{cases}
            \begin{pmatrix}
                    1 & e^{2\pi i\nu}|x-x_{n,t}^*|^{2\nu}e^{2n\phi_+(x)} \\
                    0 & 1
                \end{pmatrix}, & \mbox{for $x\in (-\infty, a_{n,t}')$,}\\[3ex]
                \begin{pmatrix}
                    e^{-2n\phi_{+}(x)}e^{-2\pi i\nu} & |x-x_{n,t}^*|^{2\nu} \\
                    0 & e^{-2n\phi_{-}(x)}e^{2\pi i\nu}
                \end{pmatrix}, & \mbox{for $x\in [a_{n,t}',b_{n,t}']$,} \\[3ex]
                \begin{pmatrix}
                    e^{-2\pi i\nu} & |x-x_{n,t}^*|^{2\nu}e^{2n\phi(x)} \\
                    0 & e^{2\pi i\nu}
                \end{pmatrix}, & \mbox{for $x\in (b_{n,t}',x_{n,t}^*)$,}\\[3ex]
                \begin{pmatrix}
                    1 & |x-x_{n,t}^*|^{2\nu}e^{2n\phi(x)} \\
                    0 & 1
                \end{pmatrix}, & \mbox{for $x\in (x_{n,t}^*,+\infty)$,}
            \end{cases}
        \end{equation}
and for $t\leq 1$,
\begin{equation} \label{RHP T: b2-2}
            v_T(x)=
            \begin{cases}
                \begin{pmatrix}
                    e^{-2n\phi_{+}(x)} & 1 \\
                    0 & e^{-2n\phi_{-}(x)}
                \end{pmatrix}, & \mbox{for $x\in (a_{n,t}',b_{n,t}')$,} \\[3ex]
               \begin{pmatrix}
                    1 & e^{2n\phi(x)} \\
                    0 & 1
                \end{pmatrix}, & \mbox{for $x\in (-\infty, a_{n,t}')\cup
                (b_{n,t}',+\infty)$.}
            \end{cases}
        \end{equation}
This way of writing down the jump matrix is convenient because it is now clear that we can factorize the jump matrix on $[a_{n,t}',b_{n,t}']$. Indeed, for $x\in[a_{n,t}', b_{n,t}']$, the jump matrix can
be written as follows,
\begin{multline}\label{vt}
v_T(z)=\begin{pmatrix}
                    1 & 0 \\
                    |x-x_{n,t}^*|^{-2\nu}e^{-2n\phi_{-}(x)}e^{2\pi i\nu} & 1
                \end{pmatrix}\begin{pmatrix}
                    0 & |x-x_{n,t}^*|^{2\nu} \\
                    -|x-x_{n,t}^*|^{-2\nu} & 0
                \end{pmatrix}\\
                \times\quad \begin{pmatrix}
                    1 & 0 \\
                    |x-x_{n,t}^*|^{-2\nu}e^{-2n\phi_{+}(x)}e^{-2\pi i\nu} & 1
                \end{pmatrix}.
\end{multline}
Because we can extend the first and the last factor analytically to the lower resp. the upper half plane,
this factorization allows us to deform the jump contour in such a way that the different factors of the jump matrix lie on different curves.
This is very convenient because it allows us to transform jumps that are oscillatory with $n$ on the real line
to jumps that are exponentially decaying on some contour in the complex plane.
We deform the jump contour to a
lens-shaped contour $\Sigma_S$ as shown in Figure \ref{figure:
lens}.

\medskip

We define
\begin{equation}\label{SinT}
S(z)=\begin{cases}
\begin{array}{ll}
T(z)&\mbox{ outside the lens-shaped region,}\\
T(z)\begin{pmatrix}
                    1 & 0 \\
                    -(x_{n,t}^*-z)^{-2\nu}e^{-2n\phi(z)}e^{-2\pi i\nu} & 1
                \end{pmatrix}&\mbox{ in the upper parts of the
                lens,}\\[3ex]
T(z)\begin{pmatrix}
                    1 & 0 \\
                    (x_{n,t}^*-z)^{-2\nu}e^{-2n\phi(z)}e^{2\pi i\nu} & 1
                \end{pmatrix}&\mbox{ in the lower parts of the
                lens.}
\end{array}
\end{cases}
\end{equation}
Due to the opening of the lens, the analytic continuations of the three factors of the jump matrix $v_T$ now live on different contours.
We obtain the following RH problem for $S$.

\subsubsection*{RH problem for $S$:}
\begin{itemize}
    \item[(a)] $S:\mathbb{C}\setminus \Sigma_S \to\mathbb{C}^{2\times 2}$ is analytic,
    \item[(b)] $S_+(x)=S_-(x)v_S(x)$ for $x\in\mathbb{R}$, with
        \begin{equation} \label{RHP S: b}
            v_S(x)=
            \begin{cases}
                \begin{pmatrix}
                    1 & e^{2\pi i\nu}|x-x_{n,t}^*|^{2\nu}e^{2n\phi_+(x)} \\
                    0 & 1
                \end{pmatrix}, & \mbox{for $x\in (-\infty, a_{n,t}')$,}\\[3ex]
                \begin{pmatrix}
                    1 & 0 \\
                    (z-x_{n,t}^*)^{-2\nu}e^{-2n\phi(x)} & 1
                \end{pmatrix}, & \mbox{for $x\in\Sigma_S\cap\mathbb C^{\pm}$,}\\[3ex]
                \begin{pmatrix}
                    0 & |x-x_{n,t}^*|^{2\nu} \\
                    -|x-x_{n,t}^*|^{-2\nu} & 0
                \end{pmatrix}, & \mbox{for $x\in
                [a_{n,t}',b_{n,t}']$,}\\[3ex]
                \begin{pmatrix}
                    e^{-2\pi i\nu} & |x-x_{n,t}^*|^{2\nu}e^{2n\phi(x)} \\
                    0 & e^{2\pi i\nu}
                \end{pmatrix}, & \mbox{for $x\in (b_{n,t}',x_{n,t}^*)$,}\\[3ex]
                \begin{pmatrix}
                    1 & |x-x_{n,t}^*|^{2\nu}e^{2n\phi(x)} \\
                    0 & 1
                \end{pmatrix}, & \mbox{for $x\in (x_{n,t}^*,+\infty)$,}\\[3ex]
            \end{cases}
        \end{equation}
    \item[(c)] $S(z)=I+\bigO(1/z)$,\qquad as $z\to\infty$,
    \item[(d)] $S(z)(z-x_{n,t}^*)^{\nu\sigma_3}$ is bounded near $x^*$.
\end{itemize}

\begin{figure}[t]
\begin{center}
    \setlength{\unitlength}{1mm}
    \begin{picture}(137.5,26)(-2.5,11.5)

        \put(25,25){\thicklines\circle*{.8}} \put(20.2,27){\small $a_{n,t}'$}
        \put(65,25){\thicklines\circle*{.8}} \put(64.8,27){\small $b_{n,t}'$}
        \put(86,25){\thicklines\circle*{.8}} \put(86,27){\small $x_{n,t}^*$}

        \put(77,25){\thicklines\vector(1,0){.0001}}
        \put(99,25){\thicklines\vector(1,0){.0001}}
        \put(9.6,25){\line(1,0){100}} \put(18,25){\thicklines\vector(1,0){.0001}}
        \put(46,25){\thicklines\vector(1,0){.0001}}

        \qbezier(25,25)(45,45)(65,25) \put(46,35){\thicklines\vector(1,0){.0001}}
        \qbezier(25,25)(45,5)(65,25) \put(46,15){\thicklines\vector(1,0){.0001}}

    \end{picture}
    \caption{The lens-shaped contour $\Sigma_S$}
    \label{figure: lens}
\end{center}
\end{figure}
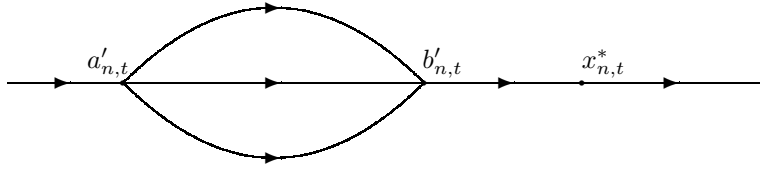

Using the fact that $S(z)=T(z)$ for $z$ outside the lens-shaped
region, and (\ref{equation g phi}), we find from (\ref{KinT}) that,
for $x,y$ in some sufficiently small fixed neighborhood of $x^*$,
\begin{equation}\label{KinS}
    K_{n,t}(x,y)
        =\frac{(x-x_{n,t}^*)_+^{\nu}(y-x_{n,t}^*)_+^{\nu}
        e^{n\phi_{+}(x)}e^{n\phi_{+}(y)}}{2\pi i(x-y)}
        \begin{pmatrix}
            0 & 1
        \end{pmatrix}
        S_{+}^{-1}(y) S_{+}(x)
        \begin{pmatrix}
            1 \\ 0
        \end{pmatrix}.
\end{equation}

\subsection{Construction of parametrices}

If we take $z$ away from arbitrary small but fixed neighborhoods surrounding $a$, $b$,
and $x^*$, it can be checked as e.g.\ in \cite{Deift} using the Cauchy-Riemann conditions that, for
$n\to\infty$, the jump matrix $v_S(z)$ converges exponentially
fast to a jump matrix $v^{(\infty)}(z)$ given by
\begin{equation}
v^{(\infty)}(x)=\begin{cases}\begin{array}{ll}
\begin{pmatrix}
0& |x-x_{n,t}^*|^{2\nu}\\
-|x-x_{n,t}^*|^{-2\nu}&0
\end{pmatrix},&\mbox{ for $x\in [a_{n,t}',b_{n,t}']$,}\\
e^{-2\pi i\nu\sigma_3},&\mbox{ for $x\in (b_{n,t}',x_{n,t}^*)$,}\\
I,&\mbox{ elsewhere.}
\end{array}
\end{cases}
\end{equation}
In other words, the jump matrices on the lips on the lens and on $\mathbb R\setminus [a-\epsilon, x^*+\epsilon]$ are exponentially close to the identity matrix. As $t>1$, on $(b+\epsilon, x^*-\epsilon)$, in the gap in between the two intervals of the spectrum, the off-diagonal entry is also exponentially small.

\medskip

If we ignore, for a moment, small neighborhoods $U_a$, $U_b$, and
$U_{x^*}$ surrounding $a$, $b$, and $x^*$, the RH problem for $S$
is reduced, up to exponentially small jumps, to a RH problem which we call the
RH problem for the outside parametrix, referring to the region away from the local neighborhoods of the special points. Besides the outside
parametrix, we will need local parametrices near $a$, $b$, and $x^*$
points in order to obtain uniform asymptotics for $S$.
The outside parametrix will determine the asymptotics for $S$ away from the special points, while the local parametrices will contribute to the local behavior of $S$. Of particular importance for us is the construction of the local parametrix near $x^*$ in Section \ref{section: local}, which will in the end describe the local behavior of the eigenvalues in the new cut.

\subsubsection{Outside parametrix}
Ignoring the exponentially small jumps and small neighborhoods of
the special points, our RH problem reduces to the following.
\subsubsection*{RH problem for $P^{(\infty)}$}
\begin{itemize}
\item[(a)] $P^{(\infty)}:\mathbb{C}\setminus [a_{t}',x_{t}^*]
\to\mathbb{C}^{2\times 2}$ is analytic,
    \item[(b)] $P^{(\infty)}$ satisfies the following jump
    conditions,
    \begin{align}&P_+^{(\infty)}(x)=P_-^{(\infty)}(x)
    \begin{pmatrix}0&|x-x_{n,t}^*|^{2\nu}\\-|x-x_{n,t}^*|^{-2\nu}&0
    \end{pmatrix},&\mbox{ as $x\in[a_{n,t}',b_{n,t}']$},\\
    &P_+^{(\infty)}(x)=P_-^{(\infty)}(x)
    e^{-2\pi i\nu\sigma_3},&\mbox{ as $x\in(b_{n,t}',x_{n,t}^*)$,}
        \end{align}
    \item[(c)] $P^{(\infty)}(z)=I+\bigO(1/z)$,\qquad as $z\to\infty$.
\end{itemize}

It is clear from the jump conditions that $P^{(\infty)}=P_{n,t}^{(\infty)}$ will have some singular behavior near $a_{n,t}'$, $b_{n,t}'$, and $x_{n,t}^*$.
Since we did not specify the required behavior near those points, the solution of this RH problem is not unique. However there is only one solution which is compatible with the local parametrices that we will construct afterwards. We now construct this solution explicitly.

\medskip
If $t\leq 1$, we have that $\nu=0$ so that $P^{(\infty)}$ is
analytic in $\mathbb C\setminus [a_{n,t}',b_{n,t}']$.
One checks directly as e.g. in \ \cite{Deift}
that
\begin{equation}
\widehat P^{(\infty)}(z):=
\begin{pmatrix}
  \frac{\beta(z)+\beta(z)^{-1}}{2} & \frac{\beta(z)-\beta(z)^{-1}}{2i} \\
  -\frac{\beta(z)-\beta(z)^{-1}}{2i} & \frac{\beta(z)+\beta(z)^{-1}}{2} \\
\end{pmatrix},
    \qquad z \in \mathbb C \setminus [a_{n,t}',b_{n,t}'],
\end{equation}
with
\begin{equation} \label{def-beta}
\beta(z)= \beta_{n,t}(z) =
\left(\frac{z-b_{n,t}'}{z-a_{n,t}'}\right)^{1/4},
    \qquad z \in \mathbb C \setminus [a_{n,t}',b_{n,t}'],
\end{equation}
is a solution of the RH problem for the outside parametrix if $t\leq 1$.

\medskip

For $t>1$, the situation is slightly more complicated because of
the additional jump on $(b_{n,t}',x_{n,t}^*)$.
This jump can be created by introducing an auxiliary scalar function
$D$, which is analytic in $\mathbb C\setminus [a_{n,t}',x_{n,t}^*]$ and has the following jumps,
\begin{equation}\label{conditions D}
\begin{array}{ll}
D_+(x)D_-(x)=|x-x_{n,t}^*|^{2\nu},&\qquad\mbox{ for
$x\in[a_{n,t}',b_{n,t}']$,}\\[1ex]
D_+(x)D_-(x)^{-1}=e^{2\pi i\nu},&\qquad\mbox{ for
$x\in[b_{n,t}',x_{n,t}^*]$.}
\end{array}
\end{equation}
In addition the limit
\begin{equation}
D_\infty=\lim_{z\to\infty}D(z)\in\mathbb R,
\end{equation}
should exist, so that we can define $P^{(\infty)}$ as follows,
\begin{equation}\label{definition Pinfty}
P^{(\infty)}(z)= D_\infty^{\sigma_3}\widehat
P^{(\infty)}(z)D(z)^{-\sigma_3}.
\end{equation}
Using (\ref{conditions D}) and the fact that
\[
\widehat P_+^{(\infty)}(x)=\widehat P_-^{(\infty)}(x)\begin{pmatrix}0&1\\-1&0\end{pmatrix},\qquad\mbox{ for $x\in(a_{n,t}',b_{n,t}')$,}\]
one verifies that indeed $P^{(\infty)}$ is a
solution of the RH problem for the outside parametrix.

\medskip

We will now construct the function $D$.
Let us first consider the function
\begin{equation}
\Phi(z)=z+\sqrt{(z-1)(z+1)} \quad \mbox{ for $z\in\mathbb
C\setminus[-1,1]$},
\end{equation}
which is the conformal mapping from $\mathbb
C\setminus[-1,1]$ to the exterior of the unit disk.
This function has the convenient property that $\Phi_+(x)\Phi_-(x)=1$ for $x\in\mathbb[-1,1]$.
If we let $F$ map $[a_{n,t}',b_{n,t}']$ to $[-1,1]$,
\begin{equation}
F(z)=\frac{z-b_{n,t}'}{b_{n,t}'-a_{n,t}'}+\frac{z-a_{n,t}'}{b_{n,t}'-a_{n,t}'},
\end{equation}
we have that $\Phi_+(F(x))\Phi_-(F(x))=1$ on $[a_{n,t}',b_{n,t}']$.

\medskip

We also need the function
\[G(z)=\exp
\left(\frac{\sqrt{(z-a_{n,t}')(z-b_{n,t}')}}{\pi}
\int_{a_{n,t}'}^{b_{n,t}'}
\frac{\log\left(x_{n,t}^*-x\right)}{\sqrt{(x-a_{n,t}')(b_{n,t}'-x)}}\
\frac{dx}{z-x}\right),\]
which is analytic in $\mathbb C\setminus [a_{n,t}',b_{n,t}']$, and using a residue argument we find that
\[G_+(x)G_-(x)=x_{n,t}^*-x, \qquad\mbox{ for $x\in [a_{n,t}',b_{n,t}']$}.\]

\medskip

Now let us write $\nu=k+\Delta$, with $k\in\mathbb N\cup \{0\}$ and $|\Delta|\leq 1/2$, so that $|\Delta|$ is the distance from $\nu$ to its nearest nonnegative integer.
Define $D$ for $z\in\mathbb C\setminus[a_{n,t}',b_{n,t}']$
in the following way,
\begin{equation}\label{def D}D(z)=(z-x_{n,t}^*)^{\Delta}\Phi(F(z))^{-\Delta} G(z)^k.
\end{equation}
 Since the branch cuts of the first two factors cancel out against each other on $(-\infty,a_{n,t}')$, $D$ is analytic in $\mathbb C\setminus [a_{n,t}',x_{n,t}^*]$, and one also checks that $D$ has a limit $D_\infty$ as $z\to\infty$. Using the jump properties of $\Phi(F)$, $G$, and $(z-x_{n,t}^*)^{\Delta}$ on $(a_{n,t}',b_{n,t}')$ and on $(b_{n,t}',x_{n,t}^*)$, one verifies that (\ref{conditions D}) is satisfied, so that $P^{(\infty)}$ solves the RH problem for the outside parametrix.

\medskip

As mentioned before, $P^{(\infty)}$ has some singular behavior near $a_{n,t}'$, $b_{n,t}'$, and $x_{n,t}^*$. We have that
\begin{align*}
& P^{(\infty)}(z)=\bigO(|z-a_{n,t}'|^{-1/4}), &\mbox{ as $z\to a_{n,t}'$,}\\
& P^{(\infty)}(z)=\bigO(|z-b_{n,t}'|^{-1/4}), &\mbox{ as $z\to b_{n,t}'$,}\\
& P^{(\infty)}(z)=\bigO(|z-x_{n,t}^*|^{-|\Delta|}), &\mbox{ as $z\to x_{n,t}^*$.}
\end{align*}
This is the 'least singular' behavior that $P^{(\infty)}$ can possibly have if it satisfies the prescribed jump conditions.

\subsubsection{Local parametrices near $a$ and $b$}\label{section Airy}
The construction of the local parametrices near $a$ and $b$ uses
a model RH problem with a solution built out of the Airy function and its derivative. This construction can be done in exactly the same way
as in \cite{Deift} for $\nu=0$ and as in \cite{KV} for
$\nu>0$. The precise construction of the parametrix is not important
for us. At this point, it is enough to know that local
parametrices in sufficiently small but fixed neighborhoods $U_a$ and $U_b$
of $a$ and $b$ (note that $a_{n,t}'$ and $b_{n,t}'$ are included in
those neighborhoods for $t$ close to $1$ and $n$ large) exist in such a way that
\begin{itemize}
\item[(a)]  $P: \overline U_a  \cup \overline U_b \setminus
\Sigma_S\to \mathbb C^{2\times 2}$ is analytic, \item[(b)] for
$z\in \Sigma_S\cap (U_a\cup U_b)$, we have $P_+(z)=P_-(z)v_S(z)$,
\item[(c)] for $z\in\partial U_a\cup\partial U_b$, if we let $n\to\infty$
and $t\to 1$, we have
\begin{equation}\label{asymptotics airy}
P(z)P^{(\infty)}(z)^{-1}=I+\bigO(n^{-1}).
\end{equation}
\end{itemize}
This matching of $P$ with $P^{(\infty)}$ can only be obtained because the behavior of $P^{(\infty)}$ near the endpoints $a_{n,t}'$ and $b_{n,t}'$ is 'not too bad'. If we would have chosen a different outside parametrix with different behavior near the endpoints, this would not be the case.

\subsubsection{Local parametrix near $x^*$}\label{section local 1}
The crucial part of the RH analysis consists of constructing a
local parametrix in a sufficiently small but fixed neighborhood $U_{x^*}$ of
the singular exterior point $x^*$. Our aim is to find a function
$P$ satisfying the following conditions.
\subsubsection*{RH problem for $P$}
\begin{itemize}
\item[(a)]$P:\overline U_{x^*}\setminus\Sigma_S\to\mathbb
C^{2\times 2}$ is analytic, \item[(b)] $P_+(z)=P_-(z)v_S(z)$ for
$z\in U_{x^*}\cap \Sigma_S$, \item[(c)] If we take the double
scaling limit where we let $n\to\infty$ and at the same time we
let $t\to 1$ in such a way that $|t-1|<M\frac{\log n}{n}$, and if we then put
\begin{equation}s=
2(t-1)\frac{n}{\log n}\ {\small\int_b^{x^*}\hspace{-0.4cm}\frac{1}{\sqrt{(s-a)(s-b)}}\,ds},\end{equation}
we have that
\begin{equation}\label{matchingP1}P(z)=P^{(\infty)}(z)\times\
\begin{cases}\begin{array}{ll}
I+\bigO(n^{-\frac{1}{2}+|\Delta|}) & \mbox{ for
$z\in\partial U_{x^*}$, if $s>0$,}\\[1.5ex]
I+\bigO(n^{-\frac{1}{2}}) & \mbox{ for $z\in\partial U_{x^*}$,
if $s\leq 0$},
\end{array}
\end{cases}
\end{equation}
\end{itemize}
where once again we have written $\nu=k+\Delta$, with $k\in\mathbb N\cup \{0\}$ and $|\Delta|\leq 1/2$.

\medskip

We postpone the construction of $P$ to Section \ref{section:
local}. We assume for now the existence of the parametrix $P$
satisfying the above conditions, and proceed with the remaining
part of the RH analysis. It is important to note that the
construction of the local parametrix near $x^*$ will only work for
our particular choices of $m_{n,t}$ and $x_{n,t}^*$ in the construction of the equilibrium measure. It is remarkable that this is the only part of the
RH analysis that would fail for arbitrary bounded $\nu$ or for an arbitrary sequence $x_{n,t}^*$ tending to $x^*$.
We note already that, in (\ref{matchingP1}), the 'matching' of $P$ with $P^{(\infty)}$ breaks down if $\Delta=\pm 1/2$, or when $\nu$ is a half integer.

\subsection{Final transformation of the RH problem}\label{section final}

We now define the function $R$ as follows,
\begin{equation}\label{def R}
R(z)= \begin{cases} \begin{array}{ll}
    S(z)P^{-1}(z), & \textrm{ for $z\in U_a\cup U_b\cup U_{x^*}$}, \\
    S(z)P^{(\infty)}(z)^{-1}, & \textrm{ for $z$ outside the disks}.
    \end{array}    \end{cases}
\end{equation}

\begin{figure}[t]
\begin{center}
    \setlength{\unitlength}{1mm}
    \begin{picture}(137.5,26)(-2.5,11.5)

        \put(25,25){\thicklines\circle*{.8}} \put(25,25){\circle{15}}
            \put(26.3,31.9){\thicklines\vector(1,0){.0001}}
        \put(65,25){\thicklines\circle*{.8}} \put(65,25){\circle{15}}
            \put(66.3,31.9){\thicklines\vector(1,0){.0001}}

        \put(90,25){\thicklines\circle*{.8}}
        \put(104,25){\thicklines\vector(1,0){.0001}}
        \put(97,25){\line(1,0){11}}
        \put(72,25){\line(1,0){11}} \put(79,25){\thicklines\vector(1,0){.0001}}
        \put(7.5,25){\line(1,0){10.5}} \put(14,25){\thicklines\vector(1,0){.0001}}

        \qbezier(29,30.78)(45,41)(61,30.78) \put(46,35.85){\thicklines\vector(1,0){.0001}}
        \qbezier(29,19.22)(45,9)(61,19.22) \put(46,14.15){\thicklines\vector(1,0){.0001}}

        \put(23.3,27){\small $a_{n,t}'$}
        \put(63,27){\small $b_{n,t}'$}
        \put(90,25){\circle{15}}
        \put(90,27){\small $x_{n,t}^*$}
        \put(91.3,31.9){\thicklines\vector(1,0){.0001}}
    \end{picture}
    \caption{The contour $\Sigma_R$ after the third and final
        transformation.}
    \label{figure: contour R}
\end{center}
\end{figure}
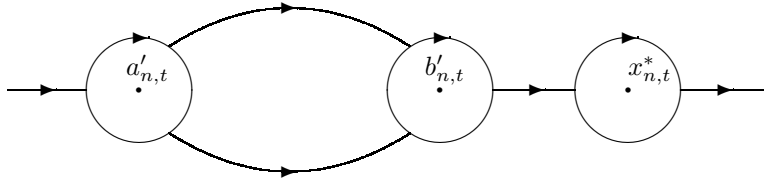

Here $P$ is the parametrix satisfying the RH problem posed in
Section \ref{section local 1} inside $U_{x^*}$, and $P$ is equal
to the Airy parametrices described in Section \ref{section Airy}
inside $U_a$ and $U_b$. Inside the disks, the parametrices are
constructed in such a way that they have exactly the same jumps as
$S$ has, and this implies that $R$ is analytic inside $U_{x^*}$,
$U_a$, and $U_b$. $R$ is also analytic on
$[a_{n,t}',b_{n,t}']\setminus (\bar U_{x^*}\cup \bar U_a\cup \bar
U_b)$ since $S$ and $P^{(\infty)}$ have the same jump here. We can
conclude that $R$ is analytic outside a contour $\Sigma_R$ as
shown in Figure \ref{figure: contour R}. Outside the disks, the jump
matrices $v_R$ converge exponentially fast to the jump
$v^{(\infty)}$ for $P^{(\infty)}$. At the boundaries of the disks, the jump matrices converge as well because of the matching of the local parametrices with the outside parametrix. One verifies that, with the contour orientated as indicated in Figure \ref{figure: contour R}, $R$ satisfies the following RH
problem.
\subsubsection*{RH problem for $R$}
\begin{itemize}
\item[(a)] $R:\mathbb C\setminus \Sigma_R\to \mathbb C^{2\times
2}$ is analytic, \item[(b)] $R_+(z)=R_-(z)v_R(z)$ for
$z\in\Sigma_R$, with $v_R$ given by
\begin{equation}\label{vR}
v_R(z)=\begin{cases}
\begin{array}{ll}
I+\bigO(e^{-cn}), &\mbox{ for $z\in\Sigma_R\cap \Sigma_S$,}\\[1.5ex]
P(z)P^{(\infty)}(z)^{-1},&\mbox{ for $z\in\partial
U_{x^*}\cup\partial U_a\cup\partial U_b$.}
\end{array}
\end{cases}
\end{equation}
\item[(c)]$R(z)=I+\bigO(z^{-1})$ as $z\to\infty$.
\end{itemize}
In the double scaling limit where we let $n\to\infty$ and
at the same time we let $t\to 0$ in such a way that
$|t-1|<M\frac{\log n}{n}$, we recall from (\ref{asymptotics airy}) and (\ref{matchingP1}) that
\begin{equation}
P(z)P^{(\infty)}(z)^{-1}=\begin{cases}
\begin{array}{ll}
I+\bigO(n^{-1}),&\mbox{ for $z\in
\partial U_a\cup\partial U_b$,}\\[1.5ex]
I+\bigO(n^{-1/2+|\Delta|}),&\mbox{ for $z\in\partial
U_{x^*}$, if $s>0$,}\\[1.5ex]
I+\bigO(n^{-1/2}),&\mbox{ for $z\in\partial U_{x^*}$, if $s\leq
0$.}
\end{array}
\end{cases}
\end{equation}
Using (\ref{vR}), we see that the jump matrix $v_R$ has the following asymptotics
in the double scaling limit.
\begin{equation}\label{vRasymptotics}
v_R(z)=\begin{cases}
\begin{array}{ll}
I+\bigO(e^{-Cn}), &\mbox{ for $z\in\Sigma_R\cap \Sigma_S$,}\\[1.5ex]
I+\bigO(n^{-1}),&\mbox{ for $z\in
\partial U_a\cup\partial U_b$,}\\[1.5ex]
I+\bigO(n^{-1/2+|\Delta|}),&\mbox{ for $z\in\partial
U_{x^*}$, if $s>0$,}\\[1.5ex]
I+\bigO(n^{-1/2})&\mbox{ for $z\in\partial U_{x^*}$, if $s\leq
0$.}
\end{array}
\end{cases}
\end{equation}
Here $C$ is an unimportant positive constant.

\medskip

From the uniform convergence of the jump matrix to the identity matrix (except when $\Delta=\pm 1/2$), it
now follows as e.g.\ in \cite{Deift, DKMVZ1} that in the double
scaling limit, also the RH solution $R$ is uniformly close to the identity
matrix,
\begin{equation}\label{asymptotics R}
R(z)=\begin{cases}\begin{array}{ll}
I+\bigO(n^{-1/2+|\Delta|}), &\mbox{ if $s>0$,}\\[1.5ex]
I+\bigO(n^{-1/2})&\mbox{ if $s\leq 0$.}
\end{array}\end{cases}
\end{equation}
In fact, one can even obtain a full asymptotic expansion in negative powers of $n$, but for our purposes, (\ref{asymptotics R}) is already sufficient.
These asymptotics hold moreover uniformly for $z\in\mathbb C\setminus \Sigma_R$, and uniformly in $s$ as well, although for $s$ near a half positive integer, (\ref{asymptotics R}) only says that $R$ is uniformly bounded, without convergence to the identity matrix.

\section{Construction of the local parametrix near
$x^*$}\label{section: local}

In this section, we construct the parametrix $P$ solving the RH
problem posed in Section \ref{section local 1} explicitly.
The
construction is different in the case where $t\leq 1$ compared to the
case where $t>1$. For $t\leq 1$, the support of $\mu_{n,t}$ consists
of one interval, and the parametrix can immediately be written
down in terms of elementary functions. For $t>1$, we will use a model RH problem built out of
Hermite polynomials. This is the crucial new feature in the RH
analysis that makes it possible to find double scaling asymptotics
near the singular exterior point.

\subsection{Construction of the parametrix if $t\leq 1$}

Recall that we are in the one-interval case with $\nu=0$ for $t\leq 1$, and that the equilibrium measure does depend on $t$ but not on $n$ in this case.
For convenience, we drop the subscript $n$ to write the functions, related to the equilibrium measure, for which it is clear that they do not depend on $n$ if $t\leq 1$.

\medskip

We can construct the parametrix $P$ directly as follows,
\begin{equation}\label{definition P 0}
P(z)=P^{(\infty)}(z)\begin{pmatrix}
1&\frac{e^{2n\phi_t(x_{t}^*)}}{2\pi i}\int_{\mathbb
R}\frac{e^{-u^2}}{u-\sqrt nf_{t}(z)}du\\0&1
\end{pmatrix}.
\end{equation}
Here we should still define $f_t$ as a conformal mapping in $U_{x^*}$, positive for $x>x_{n,t}^*$, in such a way that $P$ has the appropriate jump condition.
Using Cauchy's theorem and the fact that $P^{(\infty)}$ is analytic in $U_{x^*}$, one verifies that
\[P_+(x)=P_-(x)\begin{pmatrix}1&e^{2n\phi_{t}(x_t^*)}e^{-nf_t(x)^2}\\0&1
\end{pmatrix}, \qquad\mbox{ for $x\in U_{x^*}\cap \mathbb R$.}\]
As explained in Section \ref{section local 1}, we want to construct $P$ in such a way that
\[P_+(x)=P_-(x)v_S(x)=P_-(x)\begin{pmatrix}1&e^{2n\phi_{t}(x)}\\0&1
\end{pmatrix}, \qquad\mbox{ for $x\in U_{x^*}\cap \mathbb R$,}\]
which is now equivalent to the following condition on $f_t$,
\begin{equation}\label{condition f phi}
f_{t}(z)^2=2\phi_{t}(x_{t}^*)-2\phi_{t}(z).
\end{equation}
Since $x_t^*$ is defined as the simple zero of $h_{t}$ near
$x^*$, we have by (\ref{def phi}) that $\phi_{t}(z)-\phi_{t}(x_{t}^*)$ has a double zero at
$x_{t}^*$, and consequently $f_{t}$, defined by condition (\ref{condition f phi}), is a
conformal mapping near $x_t^*$. Furthermore we have that
\begin{equation}\label{f-x} f_{t}(x_{t}^*)=0, \qquad
f_{t}'(x_{t}^*)=\frac{1}{2^{1/4}}Q_{t}''(x_{t}^*)^{1/4}>0.
\end{equation}

\medskip

It remains tho show that the matching condition (c) in the RH
problem for $P$ is valid as well. For this purpose, we note that $\phi_t(x_t^*)\leq 0$ if $t\leq 1$.
Indeed this is the case because we know by (\ref{equation g phi}) that
\[\phi_t(x_t^*)=2g_{t}(x_t^*)-V_t(x_t^*)-\ell_t.\]
The right hand side is negative for $t<1$, since the opposite would violate the variational inequality (\ref{variationalcondition:mu-inequality}) for the equilibrium measure $\mu_t=\rho_t$.
By (\ref{definition P 0}) we now
find easily that
\begin{equation}P(z)=P^{(\infty)}(z)(I+\bigO(n^{-1/2})),
\qquad\mbox{ as $n\to\infty$, $t\leq 1$}.\label{matchingPl1}
\end{equation}
This shows that $P$ satisfies the required RH conditions, stated in Section \ref{section local 1}.

\subsection{Construction of the parametrix if $t > 1$}

Because we do not have a variational inequality for $\mu_{n,t}$ near $x^*$ any longer if $t>1$, the construction of the parametrix, as done in the previous section, fails for $t>1$. In this case we have that $\nu=s>0$.

\subsubsection{Model RH problem for $\Psi$}

In order to construct
the parametrix $P$, we need the function $\Psi=\Psi(\zeta;k)$
defined for $k\in\mathbb N\cup\{0\}$ by
\begin{equation}\label{def Psi}
\Psi(\zeta;k)={\small \begin{pmatrix}
        \frac{\pi^{1/4}\sqrt{k!}}{2^{k/2}}H_k(\zeta) &
            \displaystyle{\frac{\pi^{1/4}\sqrt{k!}}{2\cdot
             2^{k/2}\pi i}
            \int_{\mathbb R}\frac{H_k(u)
            e^{-u^2}}{u-\zeta}\,
            du}
        \\[3ex]
         -2\pi i \frac{2^{(k-1)/2}}{\pi^{1/4}\sqrt{(k-1)!}}H_{k-1}(\zeta) &
            \displaystyle{-\frac{2^{(k-1)/2}}{\pi^{1/4}\sqrt{(k-1)!}}
            \int_{\mathbb R}\frac{H_{k-1}(k u) e^{-u^2}}{u-\zeta}\,du}
    \end{pmatrix}e^{-\frac{\zeta^2}{2}\sigma_3},
    \qquad\mbox{for $\zeta\in\mathbb C \setminus\mathbb
    R$,}}
\end{equation}
where $H_k$ denotes the degree $k$ Hermite polynomial, orthonormal
with respect to the weight $e^{-x^2}$ on $\mathbb R$. It is a standard fact \cite{SaTo} that the leading coefficient of the normalized polynomial $H_k$ is equal to $\frac{2^{k/2}}{\pi^{1/4}\sqrt{k!}}$, and we agree $H_{-1}:=0$.
$\Psi=\Psi(\zeta;k)$ solves the following RH problem, which is a
slightly modified version as the RH problem for Y satisfying the
conditions (\ref{RHP Y: b})-(\ref{RHP Y: c}), but now corresponding
to the external field $x^2$ instead of $nV(x)$.

\subsubsection*{RH problem for $\Psi$}
\begin{itemize}
\item[(a)]$\Psi:\mathbb C\setminus\mathbb R\to\mathbb C^{2\times
2}$ is analytic \item[(b)] For $x\in\mathbb R$,
\begin{equation}\label{jump
Psi}\Psi_+(x)=\Psi_-(x)\begin{pmatrix}1&1\\0&1
\end{pmatrix},\end{equation}
\item[(c)]$\Psi$ behaves as follows as $\zeta\to\infty$,
\begin{equation}\label{RHP M:c}\Psi(\zeta)=\left(I
+\frac{1}{\zeta}\begin{pmatrix}0&\frac{ik!}{2^{k+1}\sqrt{\pi}}\\
-\frac{i2^k\sqrt{\pi}
}{(k-1)!}&0\end{pmatrix}+\bigO\left(\frac{1}{\zeta}\right)^2\right)\zeta^{k\sigma_3}
e^{-\frac{\zeta^2}{2}\sigma_3}, \qquad\mbox{ as
$\zeta\to\infty$}.\end{equation}
\end{itemize}
Note that the matrix on the right hand side of $(\ref{definition P 0})$ looks similar to $\Psi$ for the value of $k=0$.

\subsubsection{Construction of the parametrix}

We will define the parametrix $P$ of the following form,
\begin{equation}\label{definition: P+}
P(z)=E_{n,t}(z)\Psi(\sqrt{n}f_{n,t}(z);k)
e^{-n\phi_{n,t}(z)\sigma_3}(z-x_{n,t}^*)^{-\nu\sigma_3}.
\end{equation}
Here the parameter $k$ in the model RH problem is, as before, equal to the
non-negative integer that lies closest to $\nu$. The analytic
functions $E_{n,t}$ and $f_{n,t}$ are still to be determined at
this point.

\subsubsection*{Jump condition for $P$}

Since we assume $E_{n,t}$ analytic in $U_{x^*}$, it does not
effect the jumps of $P$.
If $f_{n,t}$ is a real conformal mapping in $U_{x^*}$ with
$f_{n,t}(x_{n,t}^*)=0$ and $f_{n,t}'(x_{n,t}^*)>0$, it follows
from the jump relation (\ref{jump Psi}) for $\Psi$ that $P$ satisfies the
jump conditions
\begin{align*}
&P_+(x)=P_-(x)\begin{pmatrix}1&
|x-x_{t}^*|^{2\nu}e^{2n\phi_{n,t}(z)}\\0&1\end{pmatrix},&\mbox{for
$x_{n,t}^*<x<x_{n,t}^*+\delta$,}\\
&P_+(x)=P_-(x)\begin{pmatrix}e^{-2\pi i\nu}&
|x-x_{n,t}^*|^{2\nu}e^{2n\phi_{n,t}(z)}\\0&e^{2\pi
i\nu}\end{pmatrix},&\mbox{for
$x_{n,t}^*-\delta<x<x_{n,t}^*$.}\\
\end{align*}
This shows already that our parametrix $P$ satisfies the jump condition we required in Section \ref{section local 1},
for any choice of $E_{n,t}$ and $f_{n,t}$. The freedom we retain to define $E_{n,t}$ and $f_{n,t}$,
will be necessary to create a suitable matching of $P$
with $P^{(\infty)}$ at $\partial U_{x^*}$.

\subsubsection*{Matching condition for $P$}
In the double scaling limit where we let $n\to\infty$ and $t\to 1$
in such a way that $|t-1|\leq M\frac{\log n}{n}$, we would like
the following matching condition to hold,
\begin{equation}
P(z)P^{(\infty)}(z)^{-1}\to I,\qquad \mbox{ for $z\in\partial
U_{x^*}$.}
\end{equation}
The value of our second scaling parameter $s=\nu$ will determine how good the matching between $P$ and $P^{(\infty)}$ can be. If $\nu$ is an integer, we will have a good matching up to order $\bigO(n^{-1/2})$. When $\nu$ moves further away from an integer, the matching will become worse, but still reasonably good. Only when $\nu$ is close to a half integer, the matching is not good anymore, and then we will only have matching up to order $\bigO(1)$.

\medskip

Let us first define $f_{n,t}$ similarly as in the case where $t\leq 1$ by
\begin{equation}\label{condition phi f}
f_{n,t}(z)^2=2\phi_{n,t}(x_{n,t}^*)-2\phi_{n,t}(z).
\end{equation}
Again we know by the definition of $x_{n,t}^*$ that $Q_{n,t}$ has a
double zero at $x_{n,t}^*$, or equivalently, that $h_{n,t}$ has a simple zero at $x_{n,t}^*$.
This implies using the definition
(\ref{def phi}) of $\phi_{n,t}$ that the right hand side in
(\ref{condition phi f}) has a double zero at $x_{n,t}^*$.
Consequently this again defines $f_{n,t}$ in a
conformal way near $x^*$, with
\begin{equation}\label{fprime}
f_{n,t}(x_{n,t}^*)=0, \qquad f_{n,t}'(x_{n,t}^*)=
\frac{1}{2^{1/4}}\, Q_{n,t}''(x_{n,t}^*)^{1/4}.
\end{equation}
As $n\to\infty$ and $t\to 1$, it is clear from the discussion in Section \ref{section equilibrium} that
\begin{equation}f_{n,t}'(x_{n,t}^*)=\frac{1}{2^{1/4}}q_V''(x^*)^{1/4}+\bigO(t-1).\label{estimate f}\end{equation}
Note also that the definition of $f_{n,t}$ ensures the existence
of a constant $C>0$ such that
\begin{equation}\label{f matching}
\sqrt n\, |f_{n,t}(z)|>C\sqrt n, \qquad \mbox{ for $z\in\partial U_{x^*}$,}
\end{equation}
under the condition that we have chosen $U_{x^*}$ sufficiently small but fixed.
This means that, as $n\to\infty$, we can use the asymptotic
condition (\ref{RHP M:c}) when evaluating $\Psi$ at $\sqrt n f_{n,t}(z)$. By (\ref{definition: P+}) we
have that
\begin{equation}P(z)=E_{n,t}(z)(I+\bigO(n^{-1/2}))
(\sqrt n f_{n,t}(z))^{k\sigma_3}e^{-n\phi_{n,t}(x_{n,t}^*)\sigma_3}
(z-x_{n,t}^*)^{-\nu\sigma_3},\qquad\mbox{ for $z\in\partial
U_{x^*}$.}\label{matchingP2}
\end{equation}
This behavior suggests how we should choose the analytic pre-factor $E_{n,t}$.
If we take $E_{n,t}$ as follows,
\begin{equation}\label{def E}
E_{n,t}(z)=P^{(\infty)}(z)(z-x_{n,t}^*)^{\nu\sigma_3}
e^{n\phi_{n,t}(x_{n,t}^*)\sigma_3}(\sqrt n f_{n,t}(z))^{-k\sigma_3},
\end{equation}
one checks directly using the definition (\ref{definition Pinfty}) of $P^{(\infty)}$ that $E_{n,t}$ is
analytic in $U_{x^*}\setminus\{x_{n,t}^*\}$ with a removable singularity at $x_{n,t}^*$.

\medskip

Inserting this definition of $E_{n,t}$ into (\ref{matchingP2}) gives us the
following behavior of $P$ for $z\in\partial
U_{x^*}$ in the double scaling limit,
\begin{multline}\label{matchingP3}
P(z)=P^{(\infty)}(z)(z-x_{n,t}^*))^{\nu\sigma_3}
e^{n\phi_{n,t}(x_{n,t}^*)\sigma_3}(\sqrt nf_{n,t}(z))^{-k\sigma_3}\\
\times \left(I
+\frac{1}{\sqrt n f_{n,t}(z)}\begin{pmatrix}0&\frac{ik!}{2^{k+1}\sqrt{\pi}}\\
-\frac{i2^k\sqrt{\pi}
}{(k-1)!}&0\end{pmatrix}+\bigO\left(\frac{1}{n}\right)
\right)\\
\times
(\sqrt nf_{n,t}(z))^{k\sigma_3}e^{-n\phi_{n,t}(x_{n,t}^*)\sigma_3}(z-x_{n,t}^*)^{-\nu\sigma_3}.
\end{multline}
We can express (\ref{matchingP3}) in the following more convenient form,
\begin{equation}\label{matchingP4}
P(z)P^{\infty}(z)^{-1}=E_{n,t}(z)\left(I +\bigO(n^{-1/2}) \right)E_{n,t}(z)^{-1}
\end{equation}
If $E_{n,t}$ would be bounded uniformly in $n$ and $t$ for $z\in U_{x^*}$, this would provide a good matching. However, for general values of $\nu$ this is not the case. We can write
\begin{equation}\label{definition Ehat}
E_{n,t}(z)=\widehat E_{n,t}(z) e^{n\phi_{n,t}(x_{n,t}^*)\sigma_3}\sqrt n^{\,-k\sigma_3},
\end{equation}
where $\widehat E_{n,t}$ is bounded uniformly in $n$ and $t$, since it depends on $n$ and $t$ only through the points $a_{n,t}$, $b_{n,t}$, and $x_{n,t}^*$, which vary smoothly with $n$ and $t$. The quality of the matching now depends on the asymptotic behavior of $e^{n\phi_{n,t}(x_{n,t}^*)\sigma_3}\sqrt n^{\,-k\sigma_3}$.
We deal with this behavior in the following proposition.
\begin{proposition}\label{proposition phix}
In the double scaling limit where $n\to\infty$ and $t\to\infty$ in such a way that
$|t-1|<M\frac{\log n}{n}$, with
\begin{equation}\label{s2}\nu=
2(t-1)\frac{n}{\log n}\ \int_b^{x^*}\hspace{-0.3cm}\frac{1}{\sqrt{(s-a)(s-b)}}\,ds\in\mathbb R^+,\end{equation}
we have that
\[e^{n\phi_{n,t}(x_{n,t}^*)\sigma_3}=\bigO(\sqrt n^{\,\nu\sigma_3}).\]
\end{proposition}
\begin{proof}
We show that \[n\phi_{n,t}(x_{n,t}^*)=\frac{\nu}{2}\log n+\bigO(1),\]
from which the proposition follows directly.
Checking the variational equality for the measure $t\mu_{n,t}$, starting from (\ref{variationalcondition:mu-equality}), learns us that $t\mu_{n,t}$ is the equilibrium measure in external field $V$, with mass $t(1-m_{n,t})$. Now it follows from a formula by Buyarov and Rakhmanov \cite{BR} that
\begin{equation}\label{BuRa}
t\mu_{n,t}-\mu_{n,1}=\int_1^{t(1-m_{n,t})}\omega_{n,s} ds,
\end{equation}
where $\omega_{n,t}$ denotes the equilibrium measure minimizing the unweighted logarithmic energy
\[I(\omega)=\iint\log\frac{1}{|x-y|}d\omega(x)d\omega(y)\] among all probability measures supported on the support $[a_{n,t}',b_{n,t}']$ of the weighted equilibrium measure $\mu_{n,t}$. In other words the equilibrium measure in an external field can be realized as an integral of unweighted equilibrium measures. Further it is known \cite{SaTo} that
\[d\omega_{n,t}(x)=\frac{1}{\pi\sqrt{(x-a_{n,t}')(b_{n,t}'-x)}}dx.\]
On the level of densities, (\ref{BuRa}) means that
\begin{multline}th_{n,t}(z)\sqrt{(z-a_{n,t}')(z-b_{n,t}')}-h_{n,1}(z)\sqrt{(z-a)(z-b)}\\=\int_1^{t(1-m_{n,t})} \frac{1}{\sqrt{(z-a_{n,s}')(z-b_{n,s}')}}ds.\end{multline}
Integrating this identity from $b_{n,t}'$ to $x_{n,t}^*$ gives us,
using the definition (\ref{def phi}) of $\phi_{n,t}$,
that \begin{multline}t\phi_{n,t}(x_{n,t}^*)-\phi_{n,1}(x_{n,t}^*)+\int_b^{b_{n,t}'}h_{n,1}(z)\sqrt{(z-a)(z-b)}dz\\
=\int_{b_{n,t}'}^{x_{n,t}^*}\int_1^{t(1-m_{n,t})} \frac{1}{\sqrt{(z-a_{n,s}')(z-b_{n,s}')}}ds\, dz.\end{multline}
Using the fact that $\phi_{n,1}(x^*)=\phi_{n,1}'(x^*)=0$ and estimating the third term on the left using the behavior of the special points
\[x_{n,t}^*=x^*+\bigO(t-1), \qquad b_{n,t}'=b+\bigO(t-1),\]
we obtain that
\[t\phi_{n,t}(x_{n,t}^*)=\int_{b_{n,t}'}^{x_{n,t}^*}\int_1^{t(1-m_{n,t})}\frac{1}{\sqrt{(x-a_{n,s}')(x-b_{n,s}')}}ds\, dx +\bigO(t-1)^{3/2},\]
Now using also the facts that
$m_{n,t}=\frac{\nu}{n}=\bigO(n^{-1})$ and that $a_{n,t}'=a+\bigO(t-1)$,
we find that
\[\phi_{n,t}(x_{n,t}^*)=(t-1) \int_b^{x^*}\frac{1}{\sqrt{(x-a)(x-b)}}dx+ \bigO(n^{-1}),\]
and consequently by (\ref{s2})
we find that
\[n\phi_{n,t}(x_{n,t}^*)=\frac{\nu}{2}\log n+\bigO(1),\]
which proves the proposition.
\end{proof}
It follows from the proposition that, in view of (\ref{definition Ehat}),
\[E_{n,t}(z)=\bigO(\sqrt n^{|\Delta|}), \qquad\mbox{ with $\Delta=\nu-k$},\]
and this gives us by (\ref{matchingP4}) a matching as follows,
\begin{equation}
\label{matchingP5}
P(z)P^{\infty}(z)^{-1}=I +\bigO(n^{-1/2+|\Delta|}).
\end{equation}
This ends the construction of the local parametrix.

\medskip

Indeed we see that the quality of the matching is good when $\nu$ is close to an integer (or when $\Delta$ is small), but is getting worse when $\nu$ approaches a half integer (or when $\Delta$ approaches $\pm 1/2$).
However, even near the half integers, the parametrix fits sufficiently well to obtain the asymptotics (\ref{asymptotics R}) for $R$, which is necessary to prove Theorem \ref{theorem: universality 2}.

\section{Universality of the eigenvalue correlation
kernel}\label{section proof} In the previous sections we found asymptotic for $R$ in
a suitable double scaling limit where $n\to\infty$ and $t\to 1$. Reversing the transformation $S\mapsto R$ and
using the explicit formula (\ref{KinS}) for the kernel $K_{n,t}$
in terms of $S$, will enable us to find asymptotics for the
eigenvalue correlation kernel and to prove Theorem \ref{theorem:
universality}.

\medskip

For $x,y\in \mathbb R \cap U_{x^*}$, we have by (\ref{def R}) that
$R(z)=S(z)P(z)^{-1}$, with $P$ the local parametrix constructed in
Section \ref{section: local}. Inserting this into (\ref{KinS}) gives
us the following identity,
\begin{equation}\label{KinPR}
    K_{n,t}(x,y) =
        \frac{(x-x_{n,t}^*)_+^{\nu}(y-x_{n,t}^*)_+^{\nu}
        e^{n\phi_{+}(x)}e^{n\phi_{+}(y)}}{2\pi i(x-y)}
        \begin{pmatrix}
            0 & 1
        \end{pmatrix}
        P_+^{-1}(y)R(y)R^{-1}(x)P_+(x)
        \begin{pmatrix}
            1 \\ 0
        \end{pmatrix}.
\end{equation}
We will now use the asymptotics (\ref{asymptotics R}) for $R$ and
the explicit formulas for $P$ in order to find asymptotics for
$K_{n,t}$.

\subsection{Proof of Theorem \ref{theorem: universality}}
We start with the proof of Theorem \ref{theorem:
universality} in the one-interval case where $t\leq 1$ and $\nu=0$.
We prove that the re-scaled eigenvalue correlation kernel is trivial in this case.

\begin{varproof}{\bf of Theorem \ref{theorem: universality} if $t\leq 1$.}
We first recall, from formula (\ref{definition P 0}), that the structure of the local parametrix is as follows,
\[P(z)=P^{(\infty)}(z)\begin{pmatrix}1&*\\0&1\end{pmatrix}, \qquad
\mbox{ for $z\in U_{x^*}$},\]
where $*$ denotes a matrix-entry for which the precise value is unimportant. Now plugging this parametrix into (\ref{KinPR}) leads to the following equation,
\begin{equation}\label{KinPR2}
    K_{n,t}(x,y) =
        \frac{e^{n\phi_{n,t,+}(x)}e^{n\phi_{n,t,+}(y)}}{2\pi i(x-y)}
        \begin{pmatrix}
            0 & 1
        \end{pmatrix}
        P^{(\infty)}(y)^{-1}R(y)R^{-1}(x)P^{(\infty)}(x)
        \begin{pmatrix}
            1 \\ 0
        \end{pmatrix}.
\end{equation}
Because of the uniform asymptotics (\ref{asymptotics R}) for $R$ and the analyticity of $R$ in the disk $U_{x^*}$, we have
that \begin{equation}
R^{-1}(y)R(x)=I+\bigO\left(\frac{x-y}{n^{1/2}}\right)\qquad \mbox{
as $x,y\to x^*$ and $n\to\infty$.}
\end{equation}
Since $P^{(\infty)}$ is analytic in $U_{x^*}$ (if $\nu=0$) and uniformly bounded in $n$ and $t$, we obtain for $n\to\infty$ and $x,y\to x^*$,
\begin{eqnarray}
    K_{n,t}(x,y) &=&
        \frac{
        e^{n\phi_{n,t,+}(x)}e^{n\phi_{n,t,+}(y)}}{2\pi i(x-y)}
        \left(\begin{pmatrix}
            0 & 1
        \end{pmatrix}
        P^{(\infty)}(y)^{-1}P^{(\infty)}(x)
        \begin{pmatrix}
            1 \\ 0
        \end{pmatrix}+\bigO\left(\frac{x-y}{n^{1/2}}\right)\right)\nonumber\\
        &=&\frac{e^{n\phi_{n,t,+}(x)}e^{n\phi_{n,t,+}(y)}}{2\pi
        i(x-y)}\ \times \ \bigO(x-y).\label{KinPinfty}
\end{eqnarray}
Now recall from (\ref{property g: 2}) (where in addition we should note that, since $\mu_t=\rho_t$ as $t\leq 1$, inequality holds near $x^*$ as well) and (\ref{equation g phi})
that $\phi_{n,t}(x)\leq 0$ for $x>b_{n,t}'$ as $t\leq 1$.
This is already sufficient to conclude that $K_{n,t}(x,y)$ is bounded, which clearly implies,
for
\[x=x^*+\frac{u}{(cn)^{1/2}}, \qquad x=x^*+\frac{v}{(cn)^{1/2}},\]
that
\begin{equation}
\frac{1}{(cn)^{1/2}}K_n(x,y)=\bigO(n^{-1/2}) \qquad \mbox{ as
$n\to\infty$.}
\end{equation}
This proves the theorem in the case $t\leq 1$.
\end{varproof}

\begin{remark}
It is worth noting that we did not use the fact that we re-scaled with a factor $(cn)^{-1/2}$. If we would re-scale by putting $x,y=x^*+\bigO(n^{-\gamma})$, for any $\gamma >0$ we obtain directly from (\ref{KinPinfty}) that $\frac{1}{n^\gamma}K_{n,t}(x,y)=\bigO(n^{-\gamma})$. The only reason why we did choose this particular scaling, is that it is the 'right' scaling for $t>1$.
\end{remark}

\medskip

For $t>1$, we have that $\nu>0$, which leads to a local parametrix which is somewhat more complicated.
This makes the proof of the theorem slightly more involved.

\begin{varproof}{\bf of Theorem \ref{theorem: universality} if
$t>1$.} We recall from Section \ref{section: local} and
(\ref{definition: P+}) in particular that the local parametrix $P$
has the form
\begin{equation}
P(z)=E_{n,t}(z)\Psi(\sqrt{n}f_{n,t}(z);k)
e^{-n\phi_{n,t}(z)\sigma_3}(z-x_{n,t}^*)^{-\nu\sigma_3}.
\end{equation}
Inserting this formula into (\ref{KinPR}) leads us to the following identity,
\begin{multline}\label{KinPR3}
    K_{n,t}(x,y) =
        \frac{1}{2\pi i(x-y)}
        \begin{pmatrix}
            0 & 1
        \end{pmatrix}
    \Psi_+^{-1}(\sqrt nf_{n,t}(y);k)E_{n,t}^{-1}(y)R(y)\\
    \times \quad R^{-1}(x)E_{n,t}(x)\Psi_+(\sqrt nf_{n,t}(x);k)
        \begin{pmatrix}
            1 \\ 0
        \end{pmatrix}.
\end{multline}
As in the case where $t<1$, we derive from the uniform asymptotics (\ref{asymptotics R}) for $R$ and the analyticity of $R$ that the following holds in the double scaling limit,
\[R^{-1}(y)R(x)=I+\bigO\left(\frac{x-y}{n^{1/2 - |\Delta|}}\right), \qquad \mbox{ for $x,y\to x^*$.}\]
The structure (\ref{definition Ehat}) of $E$ then implies together with Proposition \ref{proposition phix}
that \[E_{n,t}(z)=\bigO(\sqrt n^{\Delta\sigma_3})\widehat E_{n,t}(z),\] and this yields
\begin{equation}\label{ER}
E_{n,t}^{-1}(y)R(y)R^{-1}(x)E_{n,t}(x)=I+\bigO\left((x-y)n^{|\Delta|}\right),\qquad \mbox{ for $x,y\to x^*$.}
\end{equation}
Let us now re-scale the variables $x$ and $y$ by putting
\[x=x^*+\frac{u}{(cn)^{1/2}}, \quad y=x^*+\frac{v}{(cn)^{1/2}}, \qquad \mbox{ with }c=f_1'(x^*)^2=\frac{1}{\sqrt 2}q_V''(x^*)^{1/2},\] where $u,v\in [-M,M]$ for some arbitrary large constant $M>0$.
Since we know by (\ref{estimate f}) and from Section \ref{section equilibrium} that \[f_{n,t}'(x_{n,t}^*)^2=c+\bigO(t-1),\quad x_{n,t}^*=x^*+\bigO(t-1), \qquad\mbox{ as $n\to\infty$ and $t\to 1$,}\] we have that
\[\sqrt n f_{n,t}(x)=u+\bigO(\sqrt n(t-1)), \qquad \sqrt n f_{n,t}(y)=v+\bigO(\sqrt n(t-1)).\]
Using the above estimates, (\ref{KinPR3}) reduces to
\begin{multline}\label{KPsi}
K_n(x,y)=\frac{1}{2\pi i(x-y)}\left(\begin{pmatrix}
            0 & 1
        \end{pmatrix}
    \Psi_+^{-1}(v;k)\Psi_+(u;k)
        \begin{pmatrix}
            1 \\ 0
        \end{pmatrix}\right. \\
        \left. +\bigO\left(\frac{u-v}{n^{1/2-|\Delta|}}\right)+\bigO\left(\frac{(u-v)\log n}{n^{1/2}}\right)\right).
\end{multline}
If $\Delta$ remains a fixed distance away from $\pm 1/2$, it is now a straightforward calculation using the definition of $\Psi$ to check that
\begin{multline}\lim\frac{1}{(cn)^{1/2}}
K_{n,t}\left(x^*+\frac{u}{(cn)^{1/2}},x^*+\frac{v}{(cn)^{1/2}}\right)\\
=\sqrt\frac{k}{2}\ e^{-\frac{u^2+v^2}{2}}\ \frac{H_k(u)H_{k-1}(v)-H_k(v)H_{k-1}(u)}{u-v},
\end{multline}
which completes the proof of Theorem \ref{theorem: universality}.
\end{varproof}

\subsection{Proof of Theorem \ref{theorem: universality 2}}

In order to prove Theorem \ref{theorem: universality 2}, we need to find more accurate asymptotics for the eigenvalue correlation kernel. We will do this by exploring in more detail formula (\ref{KinPR3}). For the proof of Theorem \ref{theorem: universality}, it was sufficient to approximate $E_{n,t}^{-1}(y)R(y)R^{-1}(x)E_{n,t}(x)$ by the identity matrix. To arrive at Theorem \ref{theorem: universality 2}, we need to be a little bit more careful, and we need some better estimate, compared to (\ref{ER}).
From the asymptotic formula for $R$ we obtained in Section \ref{section final}, we recall
that
\begin{equation}
R(z)=I+\bigO(n^{-1/2+|\Delta|}),
\end{equation}
uniformly for $z\in\mathbb C\setminus \Sigma_R$,
in the double scaling limit we considered.

\medskip

Using (\ref{definition Ehat}) and Proposition \ref{proposition phix}, we can refine (\ref{ER}) in the following way as $x,y\to x^*$ in the double scaling limit,
\begin{multline}\label{ER3}
E_{n,t}^{-1}(y)R(y)R^{-1}(x)E_{n,t}(x)\\
=\begin{cases}\begin{array}{ll}
I+2\pi ic_{n,t}^+\begin{pmatrix}0&0\\1&0\end{pmatrix}(x-y)n^{\Delta}+\bigO(x-y),&\mbox{ as $\Delta\geq 0$,}\\[3ex]
I-2\pi ic_{n,t}^-\begin{pmatrix}0&1\\0&0\end{pmatrix}(x-y)n^{-\Delta}+\bigO(x-y),&\mbox{ as $\Delta\leq 0$,}\end{array}\end{cases}
\end{multline}
where $c_{n,t}^\pm$ are some sequences that are bounded in the double scaling limit.

\begin{varproof}{\bf of Theorem \ref{theorem: universality 2}}
We start by picking up the exact formula (\ref{KinPR3}) from the proof of Theorem \ref{theorem: universality}. We now follow the same calculations as in the proof of Theorem \ref{theorem: universality}, but with (\ref{ER3}) inserted into (\ref{KinPR3}) instead of the less accurate approximation we used before.
For $\Delta\geq 0$ and with
\[x=x^*+\frac{u}{(cn)^{1/2}}, \qquad x=x^*+\frac{v}{(cn)^{1/2}},\]
it is straightforward to check that (\ref{KPsi}) should now be replaced by
\begin{multline*}\label{KinPR4}
    K_{n,t}(x,y) =
        \frac{1}{2\pi i(x-y)}\left(
        \begin{pmatrix}
            0 & 1
        \end{pmatrix}
    \Psi_+^{-1}(v;k)\Psi_+(u;k)
        \begin{pmatrix}
            1 \\ 0
        \end{pmatrix}\right. \\
        +
         \left.  2\pi ic_{n,t}^+n^{\Delta}(x-y)
        \begin{pmatrix}
            0 & 1
        \end{pmatrix}
    \Psi_+^{-1}(v;k)\begin{pmatrix}0 & 0\\
1 & 0
\end{pmatrix}\Psi_+(u;k)
        \begin{pmatrix}
            1 \\ 0
        \end{pmatrix}+\bigO\left(\frac{(u-v)\log n}{n^{1/2}}\right)\right).
\end{multline*}
We have already computed the first term on the right hand side in the proof of Theorem \ref{theorem: universality}, leading to the Hermite kernel $\mathbb K^{\rm GUE}(u,v;k)$. Using (\ref{def Psi}), we can easily compute the second term as well,
\[\begin{pmatrix}
            0 & 1
        \end{pmatrix}
    \Psi_+^{-1}(v;k)\begin{pmatrix}0 & 0\\
1 & 0
\end{pmatrix}\Psi_+(u;k)
        \begin{pmatrix}
            1 \\ 0
        \end{pmatrix}=\frac{\sqrt\pi\, k!}{2^{k}}e^{-\frac{u^2+v^2}{2}}H_k(u)H_k(v),\]
which gives us the formula
\begin{multline*}\label{KinPR5}
    \frac{1}{(cn)^{1/2}}K_{n,t}(x,y) =
        \mathbb K^{\rm GUE}(u,v;k)
        +
           c_{n,t}^+\frac{\sqrt\pi k!}{2^{k}}e^{-\frac{u^2+v^2}{2}}
        H_k(u)H_k(v)\frac{1}{n^{1/2-\Delta}}+\bigO\left(\frac{\log n}{n^{1/2}}\right).
\end{multline*}
The Christoffel-Darboux formula for orthogonal polynomials allows us to rewrite the right hand side of this equation as
\begin{equation}
e^{-\frac{u^2+v^2}{2}}\sum_{j=0}^{k-1}H_j(u)H_j(v)+ c_{n,t}^+\frac{\sqrt\pi k!}{2^{k}}e^{-\frac{u^2+v^2}{2}}
        H_k(u)H_k(v)\frac{1}{n^{1/2-\Delta}}+\bigO\left(\frac{\log n}{n^{1/2}}\right),
\end{equation}
which is equal to
\begin{equation}
(1-\lambda_{n,t}^+)\mathbb K^{\rm
GUE}(u,v;k)+
\lambda_{n,t}^+\mathbb K^{\rm
GUE}(u,v;k+1)+\bigO\left(\frac{\log n}{n^{1/2}}\right),\end{equation}
 with
 \begin{equation}\label{lambda3}\lambda_{n,t}^+=c_{n,t}^+\frac{\sqrt\pi k!}{2^{k}}\frac{1}{n^{1/2-\Delta}}.\end{equation}

\medskip

For $\Delta\leq 0$, we obtain in exactly the same way that
\begin{multline}\label{KinPR6}
    \frac{1}{(cn)^{1/2}}K_{n,t}(x,y) =
        \mathbb K^{\rm GUE}(u,v;k)\\
        -
           c_{n,t}^-\pi^{3/2}\frac{2^{k+1}}{(k-1)!} e^{-\frac{u^2+v^2}{2}}
        H_{k-1}(u)H_{k-1}(v)\frac{1}{n^{1/2+\Delta}}+\bigO\left(\frac{\log n}{n^{1/2}}\right).
\end{multline}
Again using the Christoffel-Darboux formula, this can be rewritten as
\begin{equation}\label{KinPR7}
    \frac{1}{(cn)^{1/2}}K_{n,t}(x,y) =
        (1-\lambda_{n,t}^-)\mathbb K^{\rm GUE}(u,v;k)
        +\lambda_{n,t}^-\mathbb K^{\rm GUE}(u,v;k-1)+\bigO\left(\frac{\log n}{n^{1/2}}\right),\end{equation}
with \begin{equation}\label{lambda4}\lambda_{n,t}^-= c_{n,t}^-\pi^{3/2}\frac{2^{k+1}}{(k-1)!}\frac{1}{n^{1/2+\Delta}}.\end{equation}
This proves (\ref{asymptotic expansion kernel}), and the estimates (\ref{lambda1}) and (\ref{lambda2}) for $\lambda_{n,t}^\pm$ follow from (\ref{lambda3}) and (\ref{lambda4}). This completes the proof of Theorem \ref{theorem: universality 2}.
\end{varproof}

\begin{remark}
One can obtain explicit formulas for the constants $c_{n,t}^\pm$ by computing the sub-leading term in the asymptotic expansion for $R$. This would provide explicit formulas for $\lambda_{n,t}^\pm$, which are however rather complicated. Since we have not been able to simplify those expressions considerably, we feel it is not very useful to give the involved computations leading to those formulas here.
\end{remark}

\section*{Acknowledgements}

The author is grateful to Arno Kuijlaars for useful discussions and remarks.

\medskip

\noindent The author is a Postdoctoral Fellow of the Fund for Scientific Research - Flanders (Belgium), and was also supported by FWO project G.0455.04,
K.U. Leuven research grant OT/04/21,
Belgian Interuniversity Attraction Pole P06/02, and by
ESF Program MISGAM. He is also grateful to the Department of Mathematical Sciences of Brunel University West London for hospitality.

\bigskip\noindent
{\sc Department of Mathematics, Katholieke Universiteit Leuven,\\
Celestijnenlaan 200B, 3001 Leuven, Belgium}

\medskip

\noindent {\sc Department of Mathematical Sciences, Brunel University West London, \\
Uxbridge UB8 3PH
United Kingdom\\

\medskip

\noindent E-mail
address: tom.claeys@wis.kuleuven.be

\end{document}